\documentclass[aps,prx,twocolumn,superscriptaddress,showpacs,amsfonts,amsmath,floatfix,10pt]{revtex4-2}

\usepackage[table,xcdraw,dvipsnames]{xcolor}
\usepackage{amsmath, amsfonts, amssymb, amsthm, bbm}
\usepackage[normalem]{ulem}
\usepackage{dsfont}
\usepackage[T1]{fontenc}
\usepackage{enumitem}
\usepackage{xparse}
\usepackage{bm}
\usepackage{mathtools}
\usepackage{graphicx}
\usepackage{dcolumn}
\usepackage{bm}
\usepackage{times}
\usepackage{physics}
\usepackage{graphicx}
\usepackage{booktabs}
\usepackage{tikz}
\usetikzlibrary{quantikz2}
\usetikzlibrary{svg.path} 
\usepackage{xtab}
\usepackage{listings}
\usepackage{booktabs}
\usepackage{tabularx}
\newcolumntype{Y}{>{\centering\arraybackslash}X}

\usepackage{multirow}
\usepackage{colortbl}
\usepackage{pifont}
\usepackage{makecell}
\newtheorem{defi}{Definition}

\newtheorem{lemma}{Lemma}
\newtheorem{corollary}{Corollary}
\newtheorem{theorem}{Theorem}

\definecolor{myorange}{HTML}{FFC13A}
\definecolor{myblue}{HTML}{0C49AB}
\definecolor{mypurple}{HTML}{9851C5}
\definecolor{myteal}{HTML}{56DBD7}
\definecolor{mygreen}{HTML}{46C353}
\definecolor{mylightblue}{HTML}{80B0FF}
\definecolor{mylightpurple}{HTML}{D188FF}
\definecolor{mylightteal}{HTML}{85FFFB}
\definecolor{mylightgreen}{HTML}{AFFFB7}

\definecolor{orcidlogocol}{HTML}{A6CE39}
\tikzset{
  orcidlogo/.pic={
    \fill[orcidlogocol] svg{M256,128c0,70.7-57.3,128-128,128C57.3,256,0,198.7,0,128C0,57.3,57.3,0,128,0C198.7,0,256,57.3,256,128z};
    \fill[white] svg{M86.3,186.2H70.9V79.1h15.4v48.4V186.2z}
                 svg{M108.9,79.1h41.6c39.6,0,57,28.3,57,53.6c0,27.5-21.5,53.6-56.8,53.6h-41.8V79.1z M124.3,172.4h24.5c34.9,0,42.9-26.5,42.9-39.7c0-21.5-13.7-39.7-43.7-39.7h-23.7V172.4z}
                 svg{M88.7,56.8c0,5.5-4.5,10.1-10.1,10.1c-5.6,0-10.1-4.6-10.1-10.1c0-5.6,4.5-10.1,10.1-10.1C84.2,46.7,88.7,51.3,88.7,56.8z};
  }
}
\makeatletter
\newcommand{\@OrigHeightRecip}{0.00390625}

\newlength{\@curXheight}

\newcommand{\@preventExternalization}{%
\ifcsname tikz@library@external@loaded\endcsname%
\tikzset{external/export next=false}\else\fi%
}

\newcommand{\orcidlogo}{%
\texorpdfstring{%
\setlength{\@curXheight}{\fontcharht\font`X}%
\XeTeXLinkBox{%
\@preventExternalization%
\begin{tikzpicture}[yscale=-\@OrigHeightRecip*\@curXheight,
xscale=\@OrigHeightRecip*\@curXheight,transform shape]
\pic{orcidlogo};
\end{tikzpicture}%
}}{}}

\DeclareRobustCommand\orcidlinkX[3]{\href{https://orcid.org/#2}{%
\ifstrempty{#1}{}{#1\,}\orcidlogo\ifstrempty{#3}{}{\,#3}}}

\newcommand{\orcidlink}[1]{\ifpdf\orcidlinkX{}{#1}{}\fi}

\makeatother

\usepackage{xurl}
\usepackage{hyperref}
\hypersetup{breaklinks=true}

\usetikzlibrary{decorations.pathreplacing}

\NewDocumentCommand{\GateCV}{s m o}{%
  \ensuremath{\IfBooleanTF{#1}{\mediumwidehat{#2}}{\hat{#2}}\IfNoValueTF{#3}{}{\!\left(#3\right)}}%
}

\NewDocumentCommand{\GateDV}{s m o O{d}}{%
  \ensuremath{\IfBooleanTF{#1}{\mediumwidehat{#2}}{\hat{#2}}_{#4}\IfNoValueTF{#3}{}{\!\left(#3\right)}}%
}

\NewDocumentCommand{\PhaseGate}{s o O{d}}{%
  \IfBooleanTF{#1}{\GateCV{P}[#2]}{\GateDV{P}[#2][#3]}%
}
\NewDocumentCommand{\FourierTransform}{s o O{d}}{%
  \IfBooleanTF{#1}{\GateCV{F}[#2]}{\GateDV{F}[#2][#3]}%
}
\NewDocumentCommand{\SumGate}{s o O{d}}{%
  \IfBooleanTF{#1}{\GateCV{\text{\small SUM}}[#2]^{(k,l)}}{\GateDV{\mathrm{SUM}}[#2][#3]^{(k,l)}}%
}
\NewDocumentCommand{\CZGate}{s o O{d}}{%
  \IfBooleanTF{#1}{\GateCV*{\text{\small CZ}}[#2]_{kl}}{\GateDV*{\text{\small CZ}}[#2][#3]^{(k,l)}}%
}
\NewDocumentCommand{\BSGate}{s o O{d}}{%
  \IfBooleanTF{#1}{\GateCV*{\text{\small BS}}[#2]_{kl}}{\GateDV*{\text{\small BS}}[#2][#3]^{(k,l)}}%
}

\newcommand{\mediumwidehat}[1]{%
  \mkern3mu\widehat{\mkern-3mu #1 \mkern-3mu}\mkern3mu%
}

\begin{document}

\title{Equivalence of continuous- and discrete-variable gate-based quantum computers with finite energy}

\author{Alex Maltesson \orcidlink{0009-0007-3795-2100}}
\email{maltesson.alex@gmail.com}
\affiliation{Department of Microtechnology and Nanoscience (MC2), Chalmers University of Technology, SE-412 96 G\"{o}teborg, Sweden}
\author{Ludvig Rodung \orcidlink{0009-0006-0038-3285}}
\email{ludvig@rodung.se}
\thanks{Ludvig Rodung and Alex Maltesson contributed equally to this work}
\affiliation{Department of Microtechnology and Nanoscience (MC2), Chalmers University of Technology, SE-412 96 G\"{o}teborg, Sweden}
\author{Niklas Budinger \orcidlink{0009-0002-2189-0615}}
\affiliation{Center for Macroscopic Quantum States (bigQ), Department of Physics, Technical University of Denmark, 2800 Kongens Lyngby, Denmark}
\affiliation{Johannes-Gutenberg University of Mainz, Institute of Physics, Staudingerweg 7, 55128 Mainz, Germany}
\author{Giulia Ferrini \orcidlink{0000-0002-7130-6723}}
\affiliation{Department of Microtechnology and Nanoscience (MC2), Chalmers University of Technology, SE-412 96 G\"{o}teborg, Sweden}
\author{Cameron Calcluth \orcidlink{0000-0001-7654-9356}}
\email{calcluth@gmail.com}
\affiliation{Department of Microtechnology and Nanoscience (MC2), Chalmers University of Technology, SE-412 96 G\"{o}teborg, Sweden}
\affiliation{Mathematical Institute, University of Oxford, Woodstock Road, Oxford OX2 6GG, United Kingdom}

\begin{abstract}
Continuous systems are studied in many branches of modern physics, such as high-energy physics, cosmology, condensed matter physics, quantum chemistry, and field theories. Such systems are expected to benefit from the substantial advantages in computational power of quantum computers. The continuous-variable paradigm of quantum computation provides the most natural computational formalism for these tasks. However, most existing quantum hardware is based on discrete-variable systems. We address this fundamental discrepancy by providing a rigorous framework for translating native continuous-variable algorithms onto qubit-based quantum processors. This mapping is constructed from a gate-based model of continuous-variable quantum computers, consisting of states and operations built from a polynomial sequence of elementary gates in a finite set, with total energy polynomial in the number of modes. We prove that, under realistic constraints, a gate-based model of continuous-variable quantum computers can be efficiently simulated using discrete-variable devices, thereby establishing a computational equivalence between these paradigms.
\end{abstract}

\maketitle
\section{Introduction}
The simulation of continuous systems is central to many branches of modern physics, including high-energy physics~\cite{PhysRevA.92.063825,jordan2012}, cosmology~\cite{dolag2008,opanchuk2012,mocz2021}, condensed matter physics~\cite{Lamata01012018,hofstetter2018}, and quantum chemistry~\cite{Childs2022quantumsimulationof}, as well as lattice gauge theories~\cite{banuls2020,zohar2021} and scalar field theories~\cite{PhysRevA.109.052412}. Quantum computers are anticipated to offer significant computational advantages over classical computers for performing many of these simulations~\cite{lloyd1996,RevModPhys.86.153}. 
The continuous-variable (CV) paradigm of quantum computation (QC) offers a natural host for simulating continuous systems. Computation in this model relies on quantum systems with infinitely many levels, such as quantized bosonic fields, where associated quantum gates operate continuously on the states.

However, most physical quantum computing devices realized to date are based on the discrete-variable (DV) paradigm of quantum computation~\cite{preskill2018,tacchino2020,yang2023}. In this approach, computation is performed using finite-dimensional quantum systems, such as quantum bits (or qubits) or higher-dimensional equivalent, $d$-level systems called quantum-dits (or qudits), and discrete quantum gates are applied.

In light of this, a natural question that arises is: how efficiently can a continuous-variable algorithm be implemented on a discrete-variable device?  

While in Ref.~\cite{lloyd1999}, it was hypothesized that CV quantum computers are no more powerful than qubit-based devices in the computational-theoretic sense, which implies that a CV algorithm could be efficiently implemented on a DV device. However, no formal proof was given, nor was there an explanation of how one would achieve this in practice. Consequently, it remains an open question whether discrete-variable devices can efficiently simulate continuous-variable quantum computation.
\begin{figure}[b]
    \centering
    \hspace*{-0.4cm}
    \includegraphics[]{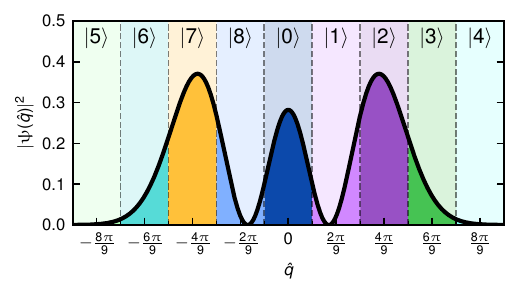}
    \vspace*{-0.4cm}
    \caption{The figure shows a visualization of how the SSD maps a two-photon Fock state to a nine-dimensional qudit. The Fock state is represented in terms of its probability density evaluated in the position basis. Each section between the grid lines corresponds to a specific level of the qudit. The areas of the colored in sections under the graph are equal to the terms in the probability distribution of the qudit, when the SSD has been applied.}
    \label{fig:SSD}
\end{figure}
\begin{figure*}[ht]
    \centering
        \input{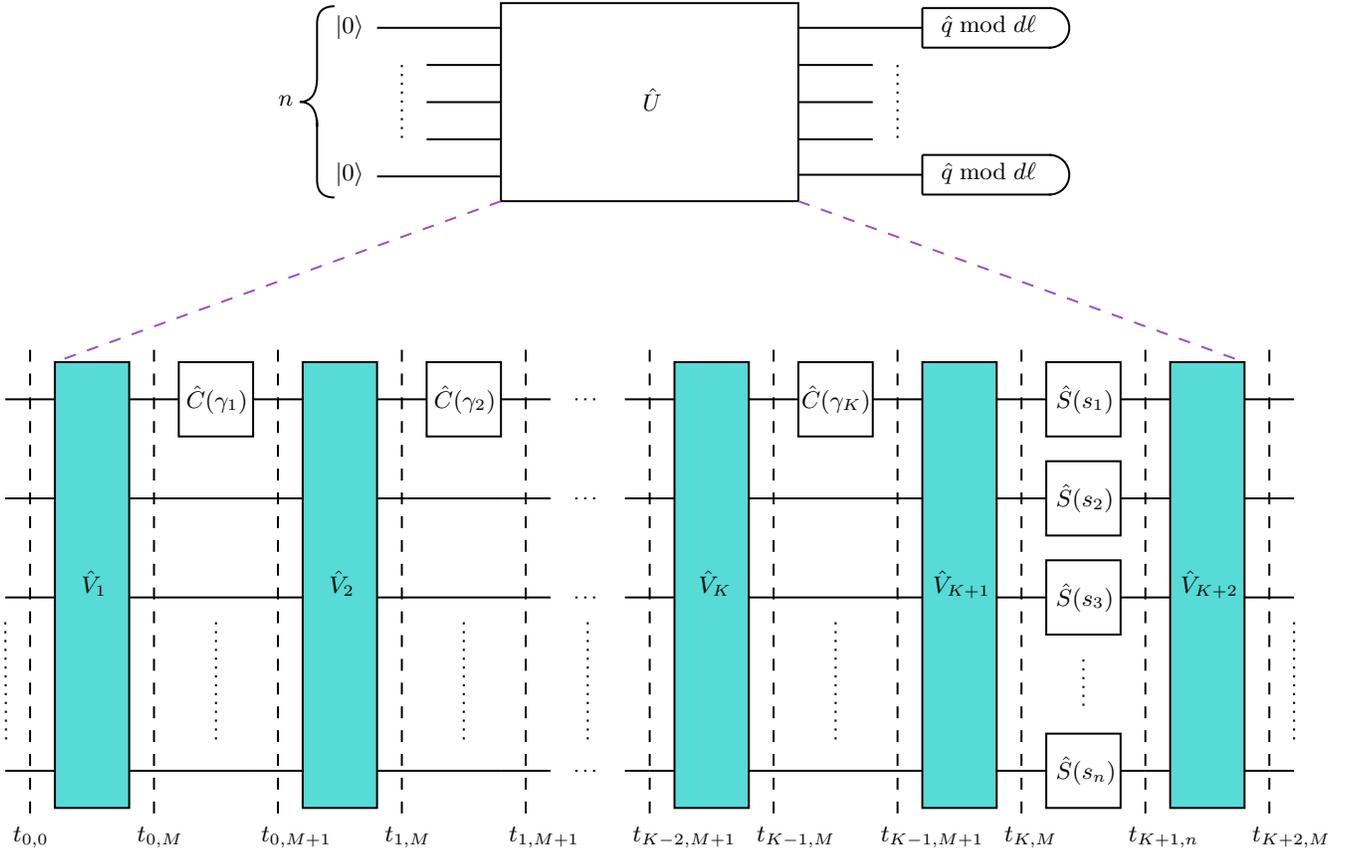}
    \caption{In this Figure, we show the circuit diagram consisting of a sequence of operations  $\hat{U}$ representing the set $\mathcal U_E$. In the top part of the figure, the diagram depicts application of the operation $\hat{U}$ on $n$ number of vacuum input states, and measurements performed in position basis modulo $d\ell$. The lower section of the figure shows the decomposition of $\hat{U}$ in terms of Eq.~\eqref{eq: RCVQC gate decomp}, where we interlace $K$ passive operations and displacements $\hat V$ with $K$ cubic phase gates $\hat{C}(\gamma)$, which is then followed by two additional $\hat V$ operations sandwiching a squeezing operation $\hat S(s)$ on each mode. In Fig.~\ref{fig:Universal_circuit_b}, we show the decomposition of the colored gates in terms of Mach-Zehnder interferometers, single-mode rotations, and displacements. The dashed black vertical lines represent time steps. We label the time steps using the notation $t_{i,j}$ where $i$ refers to the round of interlaced Gaussian and cubic phase gates, while $j$ refers to the number of gates applied during the round. Here we denote $M=n(n-1)/2+2n$.}
    \label{fig:Universal_circuit_a}
\end{figure*}
In this work, we provide an answer to this long-standing open question by providing a framework for converting a finite energy CV algorithm constructed from a polynomial number of elementary CV gates onto a DV quantum computer. We show that the approximation is efficient and has a level of precision that depends on the energy scale of the CV system. This result simultaneously proves the computational equivalence of the CV model under consideration and the DV paradigm. To do this, we utilize a tool called the stabilizer subsystem decomposition (SSD)~\cite{Shaw2024}, which allows us to map any CV state onto a DV state. As an illustrative example, Fig.~\ref{fig:SSD} shows a two-photon Fock state mapped onto a nine-dimensional qudit. Naturally, the higher the dimension of the qudit state, the higher the resolution of the mapping.

The set of operations we consider, constructed from a polynomial number of Gaussian and cubic phase gates, is given by the general class of circuits shown in Fig.~\ref{fig:Universal_circuit_a} together with Fig.~\ref{fig:Universal_circuit_b}. This set was presented as ``universal'' in Ref.~\cite{lloyd1999}, in the sense that it could be used to generate all polynomial Hamiltonians. Although a rigorous proof of this claim and a formal analysis of the error bounds in infinite dimensions remains an open problem \cite{arzani2025,chabaudprivate}, we nonetheless use this well-established and experimentally relevant class of circuits to define the model that is the central focus of this work. We refer to this model as \textit{gate-based} CVQC~\footnote{Although Ref.~\cite{lloyd1999} focuses on the set constructed from Gaussian operations in combination with access to the Kerr gate, the respective set with access to the cubic phase gate is also introduced in their work.}.

\begin{figure*}[t]
    \centering
        \input{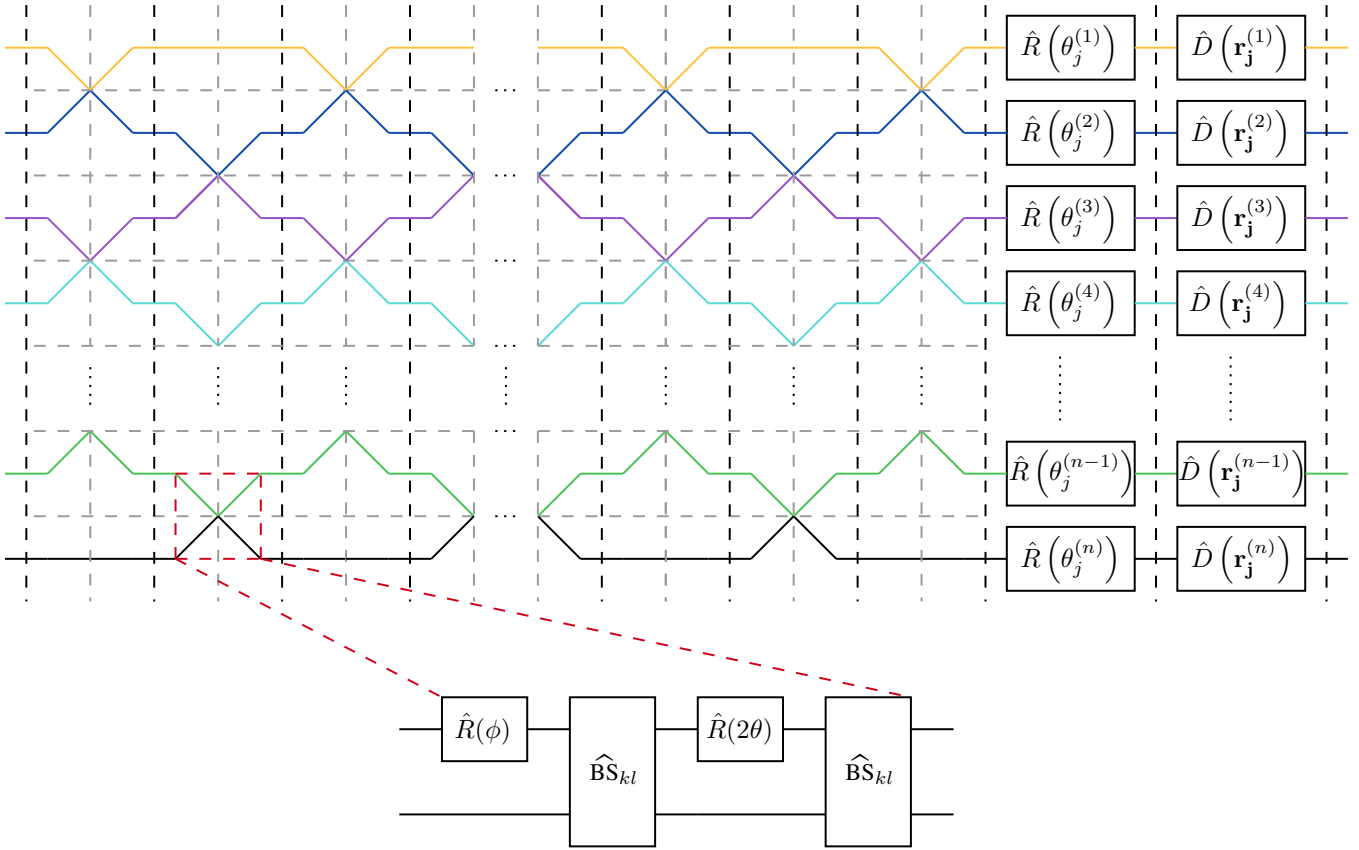}
    \caption{In this figure, we show the decomposition of an arbitrary colored gate with index $j$ from Fig.~\ref{fig:Universal_circuit_a}. This decomposition is in terms of a Mach-Zehnder interferometer network according to the design used in Ref.~\cite{clements2016}, where a total of $n(n-1)/2$ Mach-Zehnder interferometers are used with a circuit depth of $n$, which is then followed by rotations and displacements on each mode. The crossings of the wires correspond to the Mach-Zehnder interferometer from Eq.~\eqref{eq:Mach-Zehnder}, whose circuit representation in terms of two rotations and two 50:50 beam splitters is shown at the bottom part of the figure. Each combination of color and wire crossings for the Mach-Zehnder interferometer at different intersections of the grid created by gray dashed lines represents a potentially unique set of angles that parameterize the interferometer.}
    \label{fig:Universal_circuit_b}
\end{figure*}
By utilizing the SSD as a map between gate-based CVQC and DVQC, we show that the action of this generating set of CV gates can be approximated by the action of a specific set of discrete gates for qudits that have sufficiently many dimensions. By carefully analyzing the error made by the approximation, we show that it is bounded in terms of the energy of the CV system and the dimension of the corresponding qudit. Our method, therefore, provides a natural decomposition of the states, gates, and measurements from CV to qudits. Previous results have demonstrated mappings of states, operations, and measurements from CV to DV~\cite{PhysRevLett.128.110503,descamps2024,Descamps_2026,arzani2025}. Recently, Ref.~\cite{arzani2025} demonstrated for the first time a systematic way to map finite-energy CV unitary operations constructed from arbitrary polynomials to DV in a way that bounds the error of the approximation in terms of the energy of the system. However, our results focus on the gate-based CV architecture, for which we are able to bound the error in terms of the total energy of the system and the number of elementary gates applied.

We note that proving the ability to simulate CV circuits with DV circuits is sufficient to prove the equivalence of the two models. 
The CV set of gates enables the ability to engineer ``bosonic codes'' \cite{cochrane1999,gottesman2001,grimsmo2020,Brady_2024}---which encode logical DV information into CV systems---and to perform arbitrary qubit gates on them, effectively retrieving qubit quantum computation~\cite{blair2025faulttolerantquantumcomputationuniversal}. This demonstrates that discrete systems can be efficiently simulated using CV devices. Our results, therefore, complete the proof of equivalence.

A major consequence of our results is that the classical algorithms used for the simulation of DV quantum computing~\cite{vidal2003,aaronson2004,van-den-nest2010,ohliger2012,mari2012,veitch2012,pashayan2015,bravyi2016,bennink2017, bu2019, seddon2021, Marshall:23,  hahn2024b,PhysRevA.99.062337} can be applied for the simulation of CV quantum computers. This significantly advances existing efforts to simulate CV systems, which are often intractable unless the system is all Gaussian~\cite{ferraro2005,serafini2017,mari2012,veitch2013,pashayan2015} or contains inherent symmetries, such as those seen in the Gottesman-Kitaev-Preskill (GKP) encoding~\cite{garcia-alvarez2019,calcluth2022,calcluth2023,Calcluth2025}. We outline this procedure for the simulation of CV systems, which consists of using the SSD, identifying a correct dimensionality of a DV system depending on the energy of the state and operations involved, and then using the simulation algorithms for the corresponding qudit system. This result is of independent interest, given the potential and timeliness of CV systems for quantum computation and quantum error correction, and the compelling need for classical algorithms to tackle their simulation~\cite{mari2012, veitch2013, rahimi-keshari2016, liu2023simulating, oh2023tensor, PhysRevResearch.3.033018,PhysRevLett.130.090602, bourassa2021fast, hahn2024c,dias2024classical, Calcluth2025}.

We now give an outline of the remaining Sections of this paper. In Sec.~\ref{sec:main-results}, we provide an overview of the main results presented in this work, including introducing the main Theorem demonstrating an equivalence between CV and DV QCs. In Sec.~\ref{sec:background}, we provide an overview of the background of our work along with previously known results. In Sec.~\ref{sec:realistic-CVQC}, we introduce how to calculate the errors of approximating CV operations with DV operations. We also define various CV models, including those with additional restrictions, which we use to bridge the gap between CV and DV quantum computation. In Sec.~\ref{sec:equivalence}, we combine the results obtained in the previous Section that lead to the main Theorem of this work. Finally, in Sec.~\ref{sec:conclusion}, we offer some conclusions and suggest some implications of our results.

\section{Main results}
\label{sec:main-results}
The titular result of this work is to demonstrate that under realistic conditions, a gate-based CVQC can be simulated using a qudit-based quantum computer in polynomial time, with arbitrarily small error, which is bounded in terms of the dimension of the encoding qudit and the total energy of the system. Importantly, we also provide a method to perform this simulation.
These two complementary results are important from both a fundamental and a practical point of view. From a fundamental perspective, we formally demonstrate that CVQCs cannot outperform DVQCs in a way that would provide an exponential speed-up. At the same time, practically, we provide a framework to translate any CV algorithm to work on any qudit- or qubit-based hardware.

In order to achieve this, we define realistic CVQC (RCVQC) through three simple and physically motivated assumptions of CVQCs. First, we restrict the circuit such that the energy of the system is bounded throughout the evolution of the circuit. This restriction is physically motivated by the fact that, in practice, all bosonic platforms enforce a limit on the expected photon number.  Formally, we describe the energy of a multimode CV system using the Hamiltonian
\begin{align}
    \hat H=\sum_{j=1}^n\left(\hat n_j+\frac 1 2\right),
\end{align}
where $n$ is the number of modes, $\hat n$ is the number operator, and whereby $E_{\hat \rho}=\Tr(\hat H \hat \rho)$ refers to the energy of a state $\hat \rho$. 
Second, we assume that the majority of measurement values have a magnitude that is below a certain threshold, and all other measurement values are grouped together in an overflow bin. Third, we restrict the circuit such that measurements can only be resolved to a finite resolution. The second and third conditions are physically motivated since it is impossible to build a device that measures with infinite precision or infinitely large values.

Formally, we express this Theorem in terms of a general class of circuits shown in Fig.~\ref{fig:Universal_circuit_a} and Fig.~\ref{fig:Universal_circuit_b} as follows.

\begin{theorem}
    \label{theorem:main-result}
    The outcomes of a circuit representing a realistic CVQC displayed in Fig.~\ref{fig:Universal_circuit_a}, where the energy is bounded by $E_{\hat \rho}\leq E^*$ throughout the evolution of the circuit, and for which measurements are resolved to a finite (constant) resolution and can take finite values (upper bound by a constant), can be approximated by a qudit quantum computer with an error bounded by
    \begin{align}
        \epsilon \leq 1286 Kn^2 \frac{E^{*2}}{\sqrt d},
    \end{align} where $K$ is the number of rounds of interlaced Gaussian and cubic phase gates, $n$ is the number of modes, and $d$ is the dimension of each simulating qudit.
\end{theorem}
To prove this, we define two supplementary models of RCVQC with additional restrictions that make them more similar to DVQC, thus making the comparison easier. The relationship between the different models in this work and the succession of proofs we perform to compare RCVQC and DVQC is summarized in Fig.~\ref{fig:Comparison_summary}.
\begin{figure}[ht]
    \centering
    \begin{tikzpicture}[x=0.75pt,y=0.75pt,yscale=-1,xscale=1]

\draw  [fill=myorange  ,fill opacity=1 ] (130.5,35.5) .. controls (130.5,31.08) and (134.08,27.5) .. (138.5,27.5) -- (223,27.5) .. controls (227.42,27.5) and (231,31.08) .. (231,35.5) -- (231,59.5) .. controls (231,63.92) and (227.42,67.5) .. (223,67.5) -- (138.5,67.5) .. controls (134.08,67.5) and (130.5,63.92) .. (130.5,59.5) -- cycle ;
\draw  [fill=mylightblue  ,fill opacity=1 ] (131.5,105.5) .. controls (131.5,101.08) and (135.08,97.5) .. (139.5,97.5) -- (224,97.5) .. controls (228.42,97.5) and (232,101.08) .. (232,105.5) -- (232,129.5) .. controls (232,133.92) and (228.42,137.5) .. (224,137.5) -- (139.5,137.5) .. controls (135.08,137.5) and (131.5,133.92) .. (131.5,129.5) -- cycle ;
\draw  [fill=mylightpurple  ,fill opacity=1 ] (131.5,176) .. controls (131.5,171.58) and (135.08,168) .. (139.5,168) -- (224,168) .. controls (228.42,168) and (232,171.58) .. (232,176) -- (232,200) .. controls (232,204.42) and (228.42,208) .. (224,208) -- (139.5,208) .. controls (135.08,208) and (131.5,204.42) .. (131.5,200) -- cycle ;
\draw    (181,68) -- (181,95) ;
\draw [shift={(181,97)}, rotate = 270] [color={rgb, 255:red, 0; green, 0; blue, 0 }  ][line width=0.75]    (10.93,-3.29) .. controls (6.95,-1.4) and (3.31,-0.3) .. (0,0) .. controls (3.31,0.3) and (6.95,1.4) .. (10.93,3.29)   ;
\draw    (181,138.5) -- (181,165.5) ;
\draw [shift={(181,167.5)}, rotate = 270] [color={rgb, 255:red, 0; green, 0; blue, 0 }  ][line width=0.75]    (10.93,-3.29) .. controls (6.95,-1.4) and (3.31,-0.3) .. (0,0) .. controls (3.31,0.3) and (6.95,1.4) .. (10.93,3.29)   ;
\draw  [fill=mylightteal  ,fill opacity=1 ] (131.5,246) .. controls (131.5,241.58) and (135.08,238) .. (139.5,238) -- (224,238) .. controls (228.42,238) and (232,241.58) .. (232,246) -- (232,270) .. controls (232,274.42) and (228.42,278) .. (224,278) -- (139.5,278) .. controls (135.08,278) and (131.5,274.42) .. (131.5,270) -- cycle ;
\draw    (181,209.5) -- (181,236.5) ;
\draw [shift={(181,238.5)}, rotate = 270] [color={rgb, 255:red, 0; green, 0; blue, 0 }  ][line width=0.75]    (10.93,-3.29) .. controls (6.95,-1.4) and (3.31,-0.3) .. (0,0) .. controls (3.31,0.3) and (6.95,1.4) .. (10.93,3.29)   ;
\draw    (232,48) .. controls (275,48) and (279,63) .. (278,153) .. controls (277.02,241.65) and (273.12,261.41) .. (233.82,260.08) ;
\draw [shift={(232,260)}, rotate = 2.79] [color={rgb, 255:red, 0; green, 0; blue, 0 }  ][line width=0.75]    (10.93,-4.9) .. controls (6.95,-2.3) and (3.31,-0.67) .. (0,0) .. controls (3.31,0.67) and (6.95,2.3) .. (10.93,4.9)   ;

\draw (159,41.5) node [anchor=north west][inner sep=0.75pt]   [align=left] {RCVQC};
\draw (160,111.5) node [anchor=north west][inner sep=0.75pt]   [align=left] {CCVQC};
\draw (159,182) node [anchor=north west][inner sep=0.75pt]   [align=left] {MCVQC};
\draw (190,75.5) node [anchor=north west][inner sep=0.75pt]  [color=myblue  ,opacity=1 ]  {$+\epsilon _{RC}$};
\draw (190,148) node [anchor=north west][inner sep=0.75pt]  [color=myblue ,opacity=1 ]  {$=$};
\draw (163.5,252) node [anchor=north west][inner sep=0.75pt]   [align=left] {DVQC};
\draw (190,216.5) node [anchor=north west][inner sep=0.75pt]  [color=myblue ,opacity=1 ]  {$+\epsilon _{MD}$};
\draw (283,146) node [anchor=north west][inner sep=0.75pt]   [align=left] {Theorem \ref{theorem:main-result}};
\draw (121,75.5) node [anchor=north west][inner sep=0.75pt]   [align=left] {Lemma \ref{lemma:RCVQC-CCVQC}};
\draw (121,146) node [anchor=north west][inner sep=0.75pt]   [align=left] {Lemma \ref{lemma:MCVQC-CCVQC}};
\draw (121,216.5) node [anchor=north west][inner sep=0.75pt]   [align=left] {Lemma \ref{lemma:DVQC-MCVQC}};

\end{tikzpicture}
    \caption{This figure summarizes the succession of proofs we perform to compare the computational power of RCVQC and DVQC. The boxes in the figure represent the model, and an arrow linking them corresponds to the proof we perform to assess how the connected models can approximate each other. To the left of the arrows, we refer to the Lemma that states the connection.}
    \label{fig:Comparison_summary}
\end{figure}

Furthermore, by interpreting each qudit state as constructed from a series of qubits, we are able to demonstrate that qubit QCs are also capable of simulating realistic CVQCs. This is demonstrated in the following Corollary.

\begin{corollary}
    \label{corollary:qubits}
    The outcomes of RCVQC with bounded energy $E^*$ can be approximated by a \textbf{qubit} quantum computer with an error bounded by
    \begin{align}
        \epsilon \leq 1286 Kn^2 \frac{E^{*2}}{\sqrt{2^k}},
    \end{align}
    where $K$ is the number of rounds of interlaced Gaussian and cubic phase gates, $n$ is the number of modes, and $k$ is the number of qubits representing each qudit.
\end{corollary}

We present the following overview of the series of auxiliary models used to prove Theorem~\ref{theorem:main-result}. The first auxiliary model of CVQC we introduce consists of the same components as RCVQC, but the operations in this model are defined over a Hilbert space with a finite position cut-off. This cut-off also coincides with the measurement threshold introduced for RCVQC, such that everything in the overflow bin can be disregarded. We refer to this model as cut-off CVQC (CCVQC). This restriction is imposed by applying a projector that projects on a space with a finite position cut-off before every operator and measurement. Intuitively, CCVQC is structurally more similar to DVQC than RCVQC, since DVQC has an inherent maximum energy level.

The second auxiliary model of CVQC we define has the same components as CCVQC and the same cut-off, but we impose the additional restriction that the measurements are replaced with arbitrary projective \textit{modular} homodyne measurements up to finite resolution and a finite number of measurement possibilities. We call this model modular CVQC (MCVQC). The measurements are modular over some large period, which corresponds to the maximum value that can be measured in RCVQC. Furthermore, MCVQC exhibits greater similarity to DVQC, since the measurements of DVQC are intrinsically modular.

We find that it is possible to either bound or prove the equivalence of the difference between the probability functions of the CVQC models in terms of the energy of the applied operation, the number of bins chosen, and the width of the overall modular period. Practically, we find that the most likely outcomes of the measurement are bounded by the energy of the applied operations. Therefore, so long as the cut-off is large enough that measurements outside the cut-off are highly unlikely, we can instead assume the measurements are performed modulo some large number.

Furthermore, we will demonstrate that by using the SSD, it is possible to map the MCVQC model to a high-dimensional $d$ qudit model of quantum computation (DVQC). We will first prove that this map is exact when not considering operations. In other words, if we take a CV state $\hat \rho$, we find that the modular finitely-binned measurement outcome distribution corresponds exactly to that of performing a measurement of the stabilizer SSD of the state $\hat\rho$ in the computational basis.

By choosing the set of operations to be those we have defined as gate-based, we are able to make certain assumptions. First, it provides a method to track the increase in energy of the state after each sequential gate. Second, it allows us to consider the effect of each gate on the overall difference in the probability distribution between the MCVQC and DVQC models. Specifically, we use the triangle inequality to compare the difference between applying the SSD after the $j$-th gate and the $j+1$-th gate. The error between MCVQC and DVQC can then be calculated as the sum of all these small errors.

We will then use the CCVQC model to bridge the comparison of RCVQC and DVQC. First, we show that CCVQC can approximate the output statistics of a general RCVQC circuit. We make this proof by comparing the difference in statistics when applying a projection operator to the cutoff Hilbert space between each gate of a general gate-based CV circuit. Second, we then show that the CCVQC and MCVQC have the same statistics.

In summary, we are able to prove the equivalence of RCVQC and DVQC via the intermediate models CCVQC and MCVQC. By defining equivalent logical gates in DVQC, we can quantify the error between those logical gates and the CV gates in MCVQC. This allows us to calculate the error of the approximation using a difference measure between the outcome probabilities of each model. We find we can make the error of the approximation arbitrarily small $\epsilon$ by choosing $d= ( 1286 Kn^2E^{*2}/\epsilon)^2$, where $K$ is the number of rounds of interlaced Gaussian and cubic phase gates, $n$ is the number of modes, and $E^*$ is the maximum allowed energy of the system.

Furthermore, these results can immediately be extended to the case of qubits. To do so, we use the same encoding method defined for DVQC and encode qudits of dimension $d=2^k$ into $k$ qubits. Given that the dimension $d$ of an encoded qudit increases exponentially with the number of qubits $k$ for each encoded qudit, this implies that all realistically possible gate-based quantum algorithms that can be performed on CV devices can be performed on qubit QCs with a polynomial overhead in the number of modes and the total energy bound.

\section{Background}
\label{sec:background}
In the quest to simulate CVQC using DVQC, it is necessary to define a map from the first regime to the second. Although certain methods exist~\cite{pantaleoni2020,pantaleoni2021,pantaleoni2023,Shaw2024,Calcluth2024} to map CV states to DV states, they inherently face the problem that they are irreversible for the majority of input states. Throughout this work, we focus on one of these methods, namely the SSD. Despite also being irreversible, we will demonstrate that it satisfies some ideal properties for encoding CV states in DV states in a way that introduces manageable errors.

\subsection{Discrete-variable quantum computing}
\label{sec:dvqc}
Discrete-variable quantum computing is based on quantum systems that exist in a Hilbert space with finitely many dimensions. The most common discrete-variable system encountered in quantum information is the qubit, which is a two-level system. Systems with an arbitrarily large dimension $d$ are known as qudits. Operations acting on these qudits transform qudit states to other qudit states. In the following, we define two important groups of operations: the Pauli and Clifford groups. We then introduce a formal definition of the model of DVQC.

The $d$-dimensional basis states of the qudit are denoted as $\ket a$ where $a\in \mathbb Z_d$. The $d$-dimensional Pauli group is generated by the set of operators $\hat X_d=\sum_{a=0}^{d-1}\ketbra{a+1}{a}$, where addition is taken modulo $d$, and $\hat Z_d=\sum_{a=0}^{d-1}\omega_d^a\ketbra{a}{a}$, where $\omega_d=e^{2\pi i/d}$ is the $d$-th root of unity, along with a phase factor.
The single-qudit Pauli group for even $d$ is given by
\begin{align}
    \mathcal P_d = \{\zeta_d^u\hat X_d^v\hat Z_d^w : v,w \in \mathbb Z_d, u\in \mathbb Z_{2d}\}
\end{align}
where we have defined the $2d$-th root of unity as $\zeta_d=e^{\pi i/d}$,
while for odd $d$ we have
\begin{align}
    \mathcal P_d = \{\omega_d^u\hat X_d^v\hat Z_d^w : u,v,w \in \mathbb Z_d\}.
\end{align}
Pauli states refer to states that are eigenstates of Pauli operators. The Clifford group refers to the group of operations that map Pauli eigenstates to other Pauli eigenstates. Equivalently, this group of operations can be defined in terms of four generating operations: the Fourier transform, the phase gate, the controlled-$Z$ gate, and the Pauli $X$ gate.

The Fourier transform is defined as
\begin{align}\label{eq: dv Fourier transform}
    \FourierTransform = \frac{1}{\sqrt{d}} \sum_{a,b=0}^{d-1} \omega_d^{ab} \ketbra{b}{a},
\end{align}
the phase gate is defined as
\begin{align}
\label{eq: dv phase gate}
    \PhaseGate= \sum_{a=0}^{d-1} \omega^{a^2/2}_d \zeta^{ac} \ketbra{a}{a}
\end{align}
where
\begin{align}
    c=\begin{cases}
        1 \quad &\text{for odd } d \\
        0 \quad &\text{for even } d,
    \end{cases}
\end{align}
and the controlled-$Z$ gate is defined as
\begin{align} \label{eq: dv controlled-$Z$}
    \CZGate = \sum_{a,b=0}^{d-1} \omega_d^{ab} \ket{b}^{(l)} \!\prescript{(l)}{}{\bra{b}} \otimes \ket{a}^{(k)} \!\prescript{(k)}{}{\bra{a}},
\end{align}
where $k$ and $l$ refer to two qudits of a multi-qudit system.
The Clifford group can therefore be described as ${\mathcal C_d=\langle \FourierTransform, \PhaseGate,  \hat X_d,\CZGate\rangle}$, where we use the notation $\langle \cdot \rangle$ to refer to the group generated by the set.

Formally, we can introduce the model of DVQC as follows.
\begin{defi}\label{def-DVQC} (DVQC) We define the model DVQC with dimension $d$ as the set of all possible (possibly adaptive) DV quantum circuits consisting of the following components:
\begin{itemize}
\item Input states: any DV state defined on the Hilbert space $\mathbb C^{d^n}$ with $n$ a finite integer, corresponding to the number of qudits that the state spans;
\item Operations: any DV operation;
\item Arbitrary projective DV Pauli measurements.
\end{itemize}
\end{defi}

\subsection{Continuous-variable quantum computing}
\label{sec:cvqc}
In contrast to DVQC, CVQC utilizes bosonic modes that have continuous degrees of freedom. These systems are fully described in terms of the position $\hat{q}$ and momentum $\hat{p}$ quadrature operators, which satisfy the canonical commutation relation
\begin{equation}
    \left[\hat{q}_k,\hat{p}_l\right]=i\delta_{kl},
\end{equation}
where the $k$ and $l$ indices now refer to two modes of a multi-mode system. These modes are defined in the Hilbert space $L^2(\mathbb R^n)$.
We also introduce the quadrature vector as $\hat{\mathbf r}=(\hat q_1,\hat p_1,\dots,\hat q_n,\hat p_n)^T$.

As mentioned in Sec.~\ref{sec:main-results}, we refer to the energy of a state $\hat \rho$ as $E_{\hat\rho}=\Tr(\hat \rho\hat H)$. 
Note that, for a single-mode system, the energy will be
\begin{equation}
    E_{\hat{\rho}}=\langle\hat{n}\rangle_{\hat{\rho}}+\frac{1}{2}=\frac{1}{2}\left(\langle\hat{q}^2\rangle_{\hat{\rho}}+\langle\hat{p}^2\rangle_{\hat{\rho}}\right).
\end{equation}

Unlike in DVQC, the notion of universality in CVQC is difficult to define.
From a phenomenological point of view, one can define the set of operations as any operation that maps states to other valid states. However, these operations are not always physical and can be responsible for increasing the energy of a state to infinity. Instead, various methods have been suggested to reduce the set of states to those that are physically implementable or easier to parametrize. In Ref.~\cite{lloyd1999}, various different classes of operations were introduced as candidates for universality. While it was suggested that these different classes were equivalent, this remains an open problem. The first class of operations was defined as comprising $\hat U=e^{if(\hat{\mathbf r})}$ where $f$ is a polynomial function. This class was recently proven to be truly universal for all energy-constrained unitaries~\cite{arzani2025}. This means that for any, not necessarily polynomial, $f$ representing a physical unitary, there exists an approximation of the operation built from a finite sequence of gates from the first class.

For a single mode, a generating set of gate-based operations is given by displacements, the shear gate, the Fourier transform, and the cubic phase gate.
A momentum displacement is given by $\hat X(s)=e^{-i s \hat p}$, while a position displacement is given by $\hat Z(s)=e^{i s \hat q}$. By combining these two operators, we define the CV displacement operator as
\begin{align}
\label{eq: displacement}
    \hat D(\mathbf r)=e^{i\mathbf r^T\Omega \hat{\mathbf r}}
\end{align}
whereby $\Omega$ is the $2n\times 2n$ symplectic form defined as
\begin{align}
    \Omega=\bigoplus_{j=1}^n\begin{pmatrix}0&1\\-1&0\end{pmatrix},
\end{align}
where $\oplus$ refers to the direct sum and $\mathbf r=(\mathbf{r_1}^T,\dots,\mathbf{r_n}^T)^T$ with $\mathbf{r_j}=(q_j,p_j)^{T}\vphantom{T^{T^T}}\in \mathbb R^2$~\cite{ferraro2005,serafini2017}. For single-mode systems, we denote $\mathbf r=(r_q,r_p)^{T}\vphantom{T^{T^T}}$.
The shear gate (also known as the phase gate)~\cite{lloyd1999,braunstein2005} is given by
\begin{align}
\label{eq: shear gate}
    \PhaseGate*[s] = e^{is\hat q^2/2}
\end{align}
with $s\in \mathbb R$.
The cubic phase gate~\cite{lloyd1999,gottesman2001} is given by
\begin{align}
\label{eq: cubic phase gate}
    \hat{C}(\gamma)=e^{i\gamma\hat q^3},
\end{align}
where we refer to the parameter $\gamma$ as cubicity.

However, there are several other possible choices for the generators of a universal set of operations. For example, Bloch-Messiah (or Euler) decomposition~\cite{arvind1995} allows us to decompose the shear operator into a squeezing operation and two rotations, where squeezing is defined as
\begin{equation} \label{eq: squeezing gate}
    \hat{S}(r)=e^{-\frac{ir}{2}\left(\hat{q}\hat{p}+\hat{p}\hat{q}\right)},
\end{equation}
and the rotation operator as
\begin{equation} \label{eq: rotation gate}
    \hat{R}(\theta)=e^{\frac{i\theta}{
    2}\left(\hat{q}^2+\hat{p}^2\right)}.
\end{equation}
The Fourier transform is a special case of a rotation and is equal to $\hat{R}(\pi/2)$, which will look like
\begin{equation}
\label{eq: Fourier transform}
    \FourierTransform*= e^{\frac{i\pi}{4}\left(\hat q^2 + \hat p^2 \right)}.
\end{equation}
To construct universal multimode CV operations, a two-mode interaction gate is required. Both the control-$Z$ gate and beam splitter promote the set to universality in the sense of Ref.~\cite{lloyd1999}.
The controlled-$Z$ gate acting on modes $k$ and $l$ is defined as
\begin{equation} \label{eq: controlled-$Z$}
    \CZGate*(s)=e^{is\hat{q}_k \hat{q}_l}.
\end{equation}
The beam splitter defined with one variable that acts on modes $k$ and $l$ is given as
\begin{equation}
    \BSGate*(\theta)=e^{ i\theta\left(\hat{p}_k\hat{q}_l-\hat{q}_k\hat{p}_l\right)}.
\end{equation}
The special case of a $50$:$50$ beam splitter will henceforth be denoted without any parameter, and it can be constructed from the control-$Z$ gates and Fourier transforms as 
\begin{align}
    \label{eq:beam-splitter}
    \BSGate* & = \BSGate*(\pi/4).
\end{align}
Furthermore, we also generalize the notion of a beam splitter to a Mach-Zehnder interferometer, which is defined by two angles and will act on two neighboring modes \cite{clements2016}. We decompose this gate in terms of $50$:$50$ beam splitter and rotations as
\begin{equation}\label{eq:Mach-Zehnder}
    \BSGate*(\phi,\theta)=\BSGate*\hat{R}_k(2\theta)\BSGate*\hat{R}_k(\phi),
\end{equation}
Throughout the rest of this paper, when describing universal CV operations, we refer to those specified as the set of all possible gate-based operations. Specifically, this refers to the set of operations consisting of multimode Gaussian operations interlaced with cubic phase gates. In Ref.~\cite{arzani2025}, this set was proven to be equivalent to those specified by operations of the form
\begin{align}
    \hat G \hat{C}_{1}(\gamma_K)\hat V_K\dots \hat{C}_{1}(\gamma_1)\hat V_1,
\end{align}
where $\hat V_j$ are passive operations and displacements. Passive operations are Gaussian operations that do not affect the energy of the state, and are generated by Mach-Zehnder interferometers and single-mode rotations~\cite{clements2016}.

\subsection{Gottesman-Kitaev-Preskill encoding}
\label{sec:gkp}
The GKP encoding is a way to encode a qubit or a $d$-dimensional qudit into a bosonic mode in CVQC~\cite{gottesman2001}.  The code words, corresponding to qudit states $\ket{a}$ in the computational basis, are represented as a lattice consisting of Dirac combs in position and momentum space. Mathematically, they can be defined as
\begin{equation}
    \ket{a_{\text{GKP}}^{(d)}} = \sum_{n \in \mathbb{Z}} \ket{\hat{q}=\ell (nd+a)},
\end{equation}
where $\ell=\sqrt{\frac{2\pi}{d}}$ is the distance between each peak in the Dirac comb.  One of the main advantages of the GKP encoding is its inherent error-correcting properties. All displacements smaller than $\frac{\ell}{2}$ in either the position or momentum basis can be corrected using stabilizer measurements.

\subsection{Stabilizer subsystem decomposition}
\label{sec:ssd}
The SSD maps a CV state $\hat\rho$ to a qudit state of dimension $d$ utilizing the GKP encoding \cite{Shaw2024}. Intuitively, the SSD extracts the discrete logical information by projecting the continuous phase space into the unit cells defined by the GKP lattice. It is defined as the partial trace
\begin{align}
    \label{eq:ssd}
    \Tr_S(\hat \rho)=\frac{1}{\ell^n}\int_{\mathbb T_{\ell}^{2n}} \text{d} \mathbf{t}\,\hat \Pi \hat D(\mathbf t)\hat\rho  \hat D^\dagger(\mathbf t) \hat \Pi^\dagger,
\end{align}
where $\hat \Pi=\left(\sum_{a=0}^{d-1} \ketbra{a}{a_{\text{GKP}}^{(d)}}\right)^{\otimes n}$ is the projector from the GKP encoding to a corresponding qudit, and we express $\mathbb T^{m}_{h}=(-h/2,h/2]^{\times m}$. In practice, the SSD can be used to obtain discrete logical information from a CV state comprising a superposition of GKP states, while protecting against small displacement shifts. In the context of this work, the SSD is used as a method of mapping any general CV state---not only states consisting of superpositions of GKP states---onto a qudit state. As an illustrative example, we take a two photon Fock state and map it onto a nine-dimensional qudit. A visualization of this example is shown in Fig.~\ref{fig:SSD}, where the position-space probability density of the Fock state is plotted, and we see how the probability distribution of the qudit is constructed by section-wise integration. Naturally, the higher the dimension of the qudit state, the higher the resolution of the mapping.

\subsection{Norms and distance measures}
The Schatten norms generalize the notion of vector norms to operators. The general expression of the Schatten $p$-norm is given by~\cite{watrous2009,wilde2017}
\begin{align}
    \|A\|_{p} = \left(\Tr(|A|^p)\right)^{1/p},
\end{align}
where $|A|=\sqrt{A^\dagger A}$.
In the following, we will make use of the Schatten $1$-norm (i.e., the trace norm), the $2$-norm (i.e., the Frobenius norm), and the $\infty$-norm, which is the operator norm (i.e., the spectral norm). These are defined as
\begin{align}
    \|A\|_{1}&=\Tr(\sqrt{A^\dagger A}),\\
    \|A\|_{2}&=\sqrt{\Tr(A^\dagger A)}, \\
    \|A\|_{\infty}&= \sigma_{\text{max}}(A),
\end{align}
where $\sigma_{\text{max}}(A)$ represents the largest singular value of matrix $A$. The Cauchy-Schwarz inequality can be expressed in terms of these Schatten norms as
\begin{align}
\label{eq:Cauchy-Schwarz}
    \left\| XY \right\|_1 &\leq \left\| X \right\|_2 \left\| Y \right\|_2.
\end{align}

We will also make use of the total variation distance (TVD), which for two probability distributions $p,p'$ defined over a measurable space $(\Omega,\mathcal F)$ is defined as
\begin{align}
    \delta_{\text{TV}}(p,p')= \sup_{x \in \mathcal F} |p(x)-p'(x)|.
\end{align}

Furthermore, the trace distance between two operators $\hat A$ and $\hat B$ is defined as half of the trace norm of the difference between the two operators, i.e., $\frac{1}{2}\|\hat A-\hat B\|_1$~\cite{wilde2017}. This distance bounds the total variation distance of two probability functions $p_{\hat A}(\mathbf x), p_{\hat B}(\mathbf x)$ of measuring any positive operator valued measurement (POVM) element on each of the operators $\hat A$ and $\hat B$~\cite{nielsen2010}, that is,
\begin{align}
    \label{eq:tv-trace}
    \delta_{\text{TV}}(p_{\hat A},p_{\hat B})\leq \frac 1 2 \|\hat A-\hat B\|_1.
\end{align}
Therefore, for the purpose of bounding the TVD between two probability distributions, it is sufficient to bound the trace distance between two density operators.

A useful Lemma for bounding the trace distance between two operators is given by the Gentle measurement Lemma.
\begin{lemma}
\label{lemma:Gentle}
    (Gentle measurement Lemma. Lemma 9.4.1 of Ref. \cite{wilde2017}.) For a density operator $\hat \rho$ and a measurement operator $\hat\Lambda$ such that  $0<\hat \Lambda\leq I$ that has a high probability of detecting the state, i.e., $\Tr(\hat \Lambda \hat \rho)\geq 1-\epsilon$ with $\epsilon\in[0,1]$, the post-measurement state, given by $\hat \rho'= \frac{\sqrt{ \hat \Lambda }\hat\rho \sqrt{ \hat\Lambda }}{\mathrm{Tr}\hat\Lambda\hat\rho}$ has a trace distance from the original state bounded by $\sqrt{\epsilon}$, i.e.,
    \begin{align}
        \frac 1 2 \|\hat\rho - \hat\rho'\| \leq \sqrt \epsilon.
    \end{align}
\end{lemma}
With the notation $0<\hat \Lambda\leq I$, it is meant that $I - \Lambda$ is positive semidefinite and $\Lambda$ is positive definite.

\section{Realistic continuous-variable quantum computers}
\label{sec:realistic-CVQC}
In this Section, we will set up the universal set of operations for the CV system that we will consider, and introduce our models of CVQC that we will use to prove the computational equivalence of CVQC and DVQC, as well as the approximation steps that eventually lead to the sought conclusion.

\subsection{Set of operations}
\label{sec: set of operations}
We choose the set of circuits in gate-based CVQC, as consisting of a sequence of operations $\mathcal U$ specified by~\cite{upreti2025}
\begin{align}
\label{eq: RCVQC gate decomp}
    \hat V_{K+2}\hat S(\mathbf s)\hat V_{K+1} \hat C_1(\gamma_K) \hat V_{K}\dots \hat C_1(\gamma_1) \hat V_1
\end{align}
where here $\hat V_j$ are passive operations and displacements, i.e. $\hat V_j\in  \langle \BSGate*(\theta,\phi), \hat D(\mathbf r)\rangle$ is generated by a sequence of Mach-Zehnder interferometers (including single-mode rotations), and displacements. This sequence of gates is as general as choosing arbitrary Gaussian gates interlaced with cubic phase gates. This was demonstrated in Ref.~\cite{upreti2025} by showing that a cubic phase gate conjugated by a multimode Gaussian operation is equivalent to a cubic phase gate conjugated with passive operations and displacements. We define the elementary gates $\hat U_j$ as those selected from the set of single- and two-mode gates consisting of $\BSGate*(\theta,\phi), \hat D(\mathbf r), \hat S(\mathbf s)$ and $\hat C(\gamma)$. We count the number of applications of the cubic phase gate as $K$.

We define a restricted energy subset $\mathcal U_E\subset\mathcal U$ as the sequence of operations $\hat U_M\dots \hat U_1$ such that for all $j\in \{1,\dots,M\}$ we have $E_j \leq E^*$, where $\hat \rho_j$ is the state after $j$ elementary operations and $E_j$ is the energy after $j$ elementary operations. Note that the definition of this set will depend on the initial input state.

For an $n$-mode system, the energy-preserving gate $\hat V_j$ can be decomposed into $n(n-1)/2$ number of Mach-Zehnder interferometers and $n$ rotations from Eq.~(\ref{eq:Mach-Zehnder}), as is done in Ref.~\cite{clements2016}. This means that we can decompose each multimode passive operation into a polynomial sequence of single- and two-mode passive operations that do not affect the energy of the state.

\subsection{Models of continuous-variable quantum computers}
\label{sec:models}

In this Subsection, we formally define the realistic model of CVQC and the two auxiliary models that we will use to prove Theorem~\ref{theorem:main-result}, along with the corresponding probability density function (PDF) and the measurement operator for each model. A summary of these models is presented in Table~\ref{table:models}. 

\begin{defi}\label{def-RCVQC} (RCVQC) We define the RCVQC model as the set of all possible (possibly adaptive) CV quantum circuits constructed from the following components:
\begin{itemize}
    \item Input states: any CV state, with finite energy $E_{\hat \rho} \leq E^*$, defined on the Hilbert space $L^2(\mathbb R^n)$ with $n$ a finite integer, corresponding to the number of modes that the state spans;
    \item Operations: any sequence of realistic CV operations from the set $\mathcal U_E$;
    \item Arbitrary projective homodyne measurements up to finite resolution and finite number of measurement possibilities $d$ per mode, both constant with respect to the number of modes.
\end{itemize}
\end{defi}
The PDF of this circuit is represented as
\begin{align}
    \label{eq:prob-R}
    p_{R}(\mathbf{\bar x}) = \frac{1}{\mathcal N} \Tr(\hat{\tilde K}_{\mathbf{\bar x}}\hat \rho),
\end{align}
whereby
$\bar x_j \in \mathcal X=\{\ell (u-\lfloor d/2 \rfloor):u\in\mathbb Z_d\}$ and
\begin{align}
    \label{eq:kraus-main}
    \hat{\tilde K}_{\mathbf{\bar x}}=\int_{\mathbb T_{\ell}^n} \dd \mathbf s \ketbra{\hat{\mathbf q}=\mathbf{\bar x}+\mathbf s}{\hat{\mathbf q}=\mathbf{\bar x}+\mathbf s}.
\end{align}
We also define what we call the overflow bin as the measurement operator
\begin{align}
    \label{eq:kraus-overflow}
    \hat{\tilde K}_{-}=\mathbbm 1 - \sum_{\mathbf{\bar x}\in \mathcal X}\hat{\tilde K}_{\mathbf{\bar x}}
\end{align}
which can be interpreted as a Kraus operator representing the failure to measure an outcome. The PDF given above is renormalized by $\mathcal N$ which is calculated as
\begin{align}
    \mathcal N=\sum_{\mathbf{\bar x}\in\mathcal X} \Tr(\hat{\tilde K}_{\mathbf {\bar x}}\hat \rho)=1-\Tr(\hat{\tilde K}_{-}\hat \rho).
\end{align}

Next, we define a more restrictive model of CVQC. This new auxiliary model is effectively defined over a Hilbert space that has a finite position cut-off. This ensures that all measurements are within the range of measurements defined by Eq.~(\ref{eq:kraus-main}), which will render the overflow measurement operator unnecessary.
\begin{defi}\label{def-CCVQC} (CCVQC) We define the model CCVQC as the set of all possible CV quantum circuits consisting of the same input states as in Def.~\ref{def-RCVQC} but with the projector $\hat \Lambda_d$, which projects on the space of $-(\lfloor d/2\rfloor+1/2)\ell\leq \hat q \leq (\lceil d/2\rceil-1/2)\ell$, acting before every operator and measurement. We define the projector as
\begin{align}
\hat \Lambda_{d}= \left(\int_{-(\lfloor d/2\rfloor+1/2)\ell}^{(\lceil d/2\rceil-1/2)\ell} \, \text{d} x \, \ketbra{\hat{q}=x}{\hat{q}=x}\right)^{\otimes n}.
\end{align}
\end{defi}
We denote the PDF of this circuit as $p_{C}(\bar{\mathbf x})$, which has the same form as Eq.~(\ref{eq:prob-R}), but with operations that are projected in the same cutoff space. We can now describe the effective set of states throughout the evolution of the circuit as $\mathcal U_C=\{\hat \Lambda_d \hat \rho \hat \Lambda_d : \hat\rho\in \mathcal U_E\}$, and since $\hat\Lambda_d\hat {\tilde K}_{-} \hat\Lambda_d=0$, the PDF of the circuit is represented as
\begin{align}
    \label{eq:prob-C}
    p_{C}(\mathbf{\bar x}) = \Tr(\hat{\tilde K}_{\mathbf{\bar x}}\hat \rho).
\end{align}

The second auxiliary model of CVQC we will define contains the same components and cut-off as CCVQC, but with the additional restriction that the measurements are modular.
\begin{defi}\label{def-MCVQC} (MCVQC) We define the model MCVQC as the set of all possible CV quantum circuits consisting of the same input states and operations as in Def.~\ref{def-CCVQC} but with the measurements:
\begin{itemize}
    \item Arbitrary projective \textbf{modular} homodyne measurements up to finite resolution (with bin spacing $\ell=\sqrt{2\pi/d}$) and finite number of measurement possibilities $d$.
\end{itemize}
\end{defi}
The PDF of this circuit is represented as
\begin{align}
    p_{M}(\mathbf{\bar x}) = \Tr(\hat K_{\mathbf{\bar x}}\hat \rho),
\end{align}
whereby
$\mathbf{\bar x} \in \mathcal X^n$ and
\begin{align}
    \label{eq:measurement-operator}
    \hat K_{\mathbf{\bar x}}=&\sum_{\mathbf m \in \mathbb Z^{n}}\hat X(\ell d\mathbf m)\hat{\tilde K}_{\hat{\mathbf{x}}}\hat X^\dagger(\ell d\mathbf m).
\end{align}
\subsection{Approximation of continuous-variable operations in discrete variables}
\label{sec:approx-to-dv}
The key tool we use to demonstrate the simulatability of CV circuits with DV circuits is the ability to construct corresponding DV gates for every CV gate. We then calculate an upper bound on the error induced when simulating CV gates with their corresponding DV gates by mapping the action through the SSD. 

Let us first look at the gates introduced in Sec.~\ref{sec:cvqc}. For the CV Fourier transform, defined in Eq.~(\ref{eq: Fourier transform}), the DV analogue is exactly $\FourierTransform$, defined in Eq.~(\ref{eq: dv Fourier transform}). Next, we consider CV gates of the form $e^{if(\hat{\mathbf q})}$. Note that we expect the resulting phase achieved when applying each of these CV gates to the GKP basis states to be equal to the corresponding DV gate. Unless $f$ is periodic, this is not possible in general. However, we can restrict the condition to only be required for certain Dirac deltas in the Dirac comb of the GKP states. Specifically, we restrict the region such that $(\lfloor d/2\rfloor-1/2)\ell\leq\hat q\leq (\lceil d/2\rceil+1/2)\ell$. Within this region, the only support on the position basis for $\ket{a_{\text{GKP}}}$ is given by the position eigenket with eigenvalue $\ell\{a\}_d$ where $\{a\}_d=(a+\lfloor d/2 \rfloor)\mod d-\lfloor d/2 \rfloor$.

 For the CV displacement, defined in Eq.~(\ref{eq: displacement}), we begin by constructing the DV equivalent of the momentum displacement operator. We have
\begin{align}
    \hat Z_d(s)=\sum_{a=0}^{d-1} e^{is\ell \{a\}_d}\ketbra{a}{a},
\end{align}
noting that $\hat Z_d(\ell)=\hat Z_d$. The position displacement operator is defined in terms of the discrete momentum displacement operator, conjugated by Fourier transforms. 

For the shear gate, $\PhaseGate*[s]$, defined in Eq.~\eqref{eq: shear gate}, we construct a DV gate similar to the DV phase gate, defined in Eq.~(\ref{eq: dv phase gate}) but with an arbitrary phase $s$. We define this DV shear gate as 
\begin{align}
    \hat P_{d}(s)= \sum_{a=0}^{d-1} e^{is\ell\{a\}_d^2/2}  \ketbra{a}{a}.
\end{align}
We also define the DV analogue to the CV cubic phase gate, defined in Eq.~(\ref{eq: cubic phase gate}), as
\begin{align}
    \hat C_{d}(\gamma)= \sum_{a=0}^{d-1} e^{i\gamma \ell\{a\}_d^3}  \ketbra{a}{a}.
\end{align}
For the CV controlled-$Z$ gate, defined in Eq.~(\ref{eq: controlled-$Z$}), we add an arbitrary phase $s$ to the DV controlled-$Z$ gate defined in Eq.~(\ref{eq: dv controlled-$Z$}) resulting in
\begin{align}
    \CZGate(s) = \sum_{a,b=0}^{d-1} e^{is \ell \{a\}_d\{b\}_d} \ket{b}^{(l)} \!\prescript{(l)}{}{\bra{b}} \otimes \ket{a}^{(k)}\!\prescript{(k)}{}{\bra{a}}.
\end{align}

In general the SSD of a CV gate $\hat U$ acting on a state $\hat \rho$ with energy $E_{\hat \rho}$, can be expressed as $\mathrm{Tr}_S (\hat U \hat\rho \hat U^\dagger)$, while the SSD of the state, followed by performing the corresponding DV gate $\hat U_d$ can be expressed as $\hat U_d \mathrm{Tr}_S (\hat\rho) \hat U_d^\dagger$. Hence the the trace distance when approximating $\mathrm{Tr}_S (\hat U \hat\rho \hat U^\dagger)$ with $\hat U_d \mathrm{Tr}_S (\hat\rho) \hat U_d^\dagger$ can be expressed as:
\begin{align}
    \frac{1}{2}\left\| \hat U_d \mathrm{Tr}_S (\hat\rho) \hat U_d^\dagger - \mathrm{Tr}_S (\hat U \hat\rho \hat U^\dagger) \right\|_1.
\end{align}
We can now calculate this error bound for different gates. In order to bound the trace distance of performing operations before and after the SSD, we evaluate the maximum possible parameters that different CV operations can have. To do so, in Appendix~\ref{sec:appendix-energy}, we analyze the effect that various CV operations have on the energy of the system. Then, in Appendix~\ref{sec:Max-energy}, we bound the parameters of the elementary operations by calculating the worst-case maximum possible value of each parameter that could be applied to an arbitrary finite energy state while maintaining energy below the threshold $E^*$. We find that the maximum value of squeezing that can be applied to any input state before breaching the energy constraint is at most $e^r=\sqrt{2E^*}$. The maximum value of cubicity of the cubic phase gate is $8E^{*3/2}$.

We formulate the result in terms of Lemmas, with the proofs of these Lemmas provided in Appendix \ref{appendix: errors}.

The Fourier transform is invariant under the SSD, so it follows that the error will be zero.
\begin{lemma}
    \label{lemma:error-Fourier-transform}
    The result of the SSD of the CV Fourier transform $\hat F$ acting on states $\hat \rho$ with energy $E_{\hat \rho}$, can be obtained exactly by performing the SSD on the state, followed by performing the DV Fourier transform $\hat F_d$.
\end{lemma}

Similarly to the Fourier transform, we find that displacement is also invariant under the SSD, resulting in an error of zero.
\begin{table*}[ht]
    \centering
    \renewcommand{\arraystretch}{1.3}
    \setlength{\tabcolsep}{4pt}
    \begin{tabular}{|c|c|c|c|c|}
        \hline
        & RCVQC & CCVQC & MCVQC & DVQC \\
        \hline
        States &
        \multicolumn{3}{c|}{
          \begin{minipage}[c][1.1cm][c]{3.5cm}
              \centering
              $\hat \rho \in L^{2n}(\mathbb R)$
          \end{minipage}
        } &
        \begin{minipage}[c][1.1cm][c]{3.5cm}
            \centering
            $\hat \rho \in \mathbb C^{dn}$
        \end{minipage} \\
        \hline
        Operations &
        \begin{minipage}[c][1.1cm][c]{3.5cm}
            \centering
            $\hat{\tilde U} \in \mathcal U_{E}$
        \end{minipage} &
        \multicolumn{2}{c|}{
          \begin{minipage}[c][1.1cm][c]{3.5cm}
              \centering
              $\hat U \in \mathcal U_{C}$
          \end{minipage}
        } &
        \begin{minipage}[c][1.1cm][c]{3.5cm}
            \centering
            $\hat U_d$
        \end{minipage} \\
        \hline
        Measurements &
        \begin{minipage}[c][1.1cm][c]{3.5cm}
            \centering
            $\hat{\tilde K}_{\mathbf{\bar x}}, \hat{\tilde K}_-$
        \end{minipage} &
        \begin{minipage}[c][1.1cm][c]{3.5cm}
            \centering
            $\hat{\tilde K}_{\mathbf{\bar x}}$
        \end{minipage} &
        \begin{minipage}[c][1.1cm][c]{3.5cm}
            \centering
            $\hat K_{\bar{\mathbf x}}$
        \end{minipage} &
        \begin{minipage}[c][1.1cm][c]{3.5cm}
            \centering
            $\ketbra{\mathbf a}{\mathbf a}$
        \end{minipage} \\
        \hline
        \shortstack{\text{Measurement} \\ \text{outcomes}\vspace{0.95cm}}  &
        \includegraphics{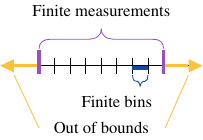} &
    \includegraphics{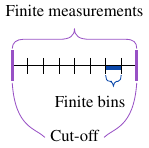} &
    \includegraphics{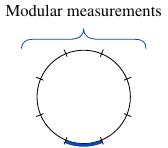} &
    \includegraphics{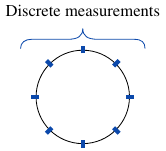}\\
        \hline
    \end{tabular}
\caption{Table highlighting the differences and similarities between each of the four models defined in this work. Here, we denote the set of operations of RCVQC as $\mathcal U_E$ representing the sequence of operations for which the energy never increases beyond $E^*$. The set of operations of CCVQC and MCVQC is represented as $\mathcal U_C$, which contains the sequence of operations in $\mathcal U_E$ interlaced with cut-off projection operators. The measurement operators $\hat {\tilde K}_{\mathbf x}$ are the POVM operators corresponding to bins in position space, $\hat {\tilde K}_{-}$ represents an ``out of bounds'' (i.e. overflow) measurement operator and $\hat {K}_{\mathbf x}$ represents modular measurement operators.}
\label{table:models}
\end{table*}
\begin{lemma}
    \label{lemma:error-displacement}
    The result of the SSD of the CV displacement $\hat D(\mathbf r)$ acting on states $\hat \rho$ with energy $E_{\hat \rho}$ can be obtained exactly by performing the SSD on the state, followed by performing the DV displacement operators.
\end{lemma}
This is shown by first comparing displacement only in the momentum quadrature and then using the Fourier transform and Lemma~\ref{lemma:error-Fourier-transform} to recover the general displacement.

The shear gate is not invariant under the SSD, leading to an error when comparing the effects of the CV and the DV shear gate through the SSD. However, this error can be bound in terms of the energy of the state, the amount of shearing, and the dimension of the DV system.
\begin{lemma}
    \label{lemma:error-shear-gate}
    The result of the SSD of the CV shear gate $\PhaseGate*[s]$ with arbitrary phase $s$ acting on states $\hat \rho$ with energy $E_{\hat \rho}$, can be approximated by performing the SSD on the state, followed by performing the DV shear gate $\PhaseGate[s]$ with a trace distance bounded by $|s| \frac{\pi}{d}  \sqrt{E_{\hat \rho}}/\sqrt 2$.
\end{lemma}
To prove this Lemma, we have used commutation relations between the shear gate and displacement, as well as with the GKP projector. We have also used the fact that we can bound the difference between a state and a slightly evolved state in terms of the commutator between the state and the Hamiltonian.

Similarly, there will be an error for the CV cubic phase gate, also bound in terms of the energy of the state, the cubic phase parameter, and the dimension of the DV system.
\begin{lemma}
    \label{lemma:error-cubic-phase}
    The result of the SSD of the CV cubic phase gate $\hat C(\gamma)$ acting on states $\hat \rho$ with energy $E_{\hat \rho}$, in the CCVQC model, can be approximated by performing the SSD on the state, followed by performing the DV cubic phase gate $\hat C_d(\gamma)$ with trace distance bounded by $5|\gamma|\pi^{3/2} \sqrt{E_{\hat \rho}/d}$.
\end{lemma}
The proof of this Lemma follows the same structure as the proof of the shear gate error, but for a different Hamiltonian. Given the maximum cubicity bounded by $8E^{*3/2}$, we see that the trace distance, independent of cubicity, is upper bounded by
\begin{align}
    223 E^{*2}/\sqrt{d}.
\end{align}

As for the shear gate and the cubic phase gate, we also get an error for the two-mode controlled-$Z$ gate bound in terms of the energy of the state, the phase parameter, and the dimension of the DV system.
\begin{lemma}
    \label{lemma:error-controlled-z}
    The result of the SSD of the CV controlled-$Z$ gate $\CZGate*(s)$ acting on states $\hat \rho$ with energy $E_{\hat \rho}$, can be approximated by performing the SSD on the state, followed by performing the DV controlled-$Z$ gate $\CZGate(s)$ with a trace distance bound by $ 2 |s| \frac{\pi^{3/2}}{d^{3/2}} \sqrt{E_{\hat \rho}}$.
\end{lemma}
Again, the proof of this Lemma follows the same structure as the proof of the shear gate error, but for a different Hamiltonian, now over two modes.

We can also obtain an error for rotation.
\begin{lemma}
\label{lemma:error-rotation}
    The result of the SSD of the rotation gate $\hat R(\theta)$ with any angle $\theta$ acting on states $\hat \rho$ with energy $E_{\hat \rho}$, can be approximated by performing the SSD on the state, followed by performing an equivalent DV gate with an error bound by $52\frac{\sqrt{E_{\hat \rho}}}{d }$.
\end{lemma}
To prove this error bound, we use the fact that rotation can be decomposed into Fourier transforms and shear gates, in conjunction with Lemma~\ref{lemma:error-Fourier-transform} and Lemma~\ref{lemma:error-shear-gate} as well as the bound on the energy of the state after applying a shear gate.

Being interested in the set specified by Eq.~(\ref{eq: RCVQC gate decomp}), in addition to displacement and the cubic phase gate, we also need the error for the Mach-Zehnder interferometer and squeezing.
\begin{lemma}
\label{lemma:error-beamsplitter}
    The result of the SSD of the Mach-Zehnder interferometer acting on states $\hat \rho$ with energy $E_{\hat \rho}$, can be approximated by performing the SSD on the state, followed by performing an equivalent DV gate with a trace distance of at most $108\frac{\sqrt{E_{\hat \rho}}}{d}$.
\end{lemma}
To prove this error bound, we have used the fact that the Mach-Zehnder interferometer can be decomposed into Fourier transforms, controlled-$Z$ gates, and rotation, in conjunction with Lemma~\ref{lemma:error-Fourier-transform} and Lemma~\ref{lemma:error-controlled-z} as well as the bound on the energy of a state after applying a controlled-$Z$ gate.
\begin{lemma}
\label{lemma:error-squeezing}
    The result of the SSD of the squeezing operator with squeezing parameter $s$ acting on states $\hat \rho$ with energy $E_{\hat \rho}$, can be approximated by performing the SSD on the state, followed by performing an equivalent DV gate with a trace distance bounded by $7e^{2r}\frac \pi d \sqrt{\frac{E_{\hat \rho}}{2}}$.
\end{lemma}
Finally, to prove this error bound, we have used the fact that squeezing can be decomposed into Fourier transforms and shear gates, in conjunction with Lemma~\ref{lemma:error-Fourier-transform}, Lemma~\ref{lemma:error-shear-gate}, and the bound on the energy of a state after applying a shear gate. As discussed above, it is also possible to bound the parameter of squeezing in terms of the total energy of the system, and therefore, we can identify the bound on the trace distance independent of the squeezing parameter as
\begin{align}
    32{E^{*3/2}}/d.
\end{align}

\subsection{Approximation of circuits}
Following the framework summarized in Fig.~\ref{fig:Comparison_summary}, we will, in this Subsection, present and prove the necessary Lemmas that ultimately allow us to compare the computational power of RCVQC with DVQC. There are two main aspects to consider when comparing the models: (i) the similarity of the states across models, and (ii) how well the PDFs of the models can be approximated by each other. 

However, it is not always necessary to rigorously evaluate both aspects for every model comparison. For example, the PDFs of the models RCVQC and CCVQC have the same form and will thus produce the same statistics. Similarly, both CCVQC and MCVQC are defined with the same operators and thus will be able to produce the same states.  

Nevertheless, we will now present the necessary Lemmas to bound the error of approximating RVCQC with DVQC. We will start from the top of Fig.~\ref{fig:Comparison_summary} and then proceed by moving downwards. 
\begin{lemma}
    \label{lemma:RCVQC-CCVQC}
    CCVQC approximates RCVQC with a TVD $\epsilon_{RC}\leq 4(L+1)\sqrt{\frac{E^*}{d\pi}}$ in terms of the maximum energy of the state throughout its evolution $E^*$ by sequences of operations selected from $\mathcal U_{E}$, the number of elementary CV gates applied $L$, and the dimension of the qudits.
\end{lemma}
\begin{proof}
We compare the set of operations $\hat{\tilde U}=\hat{\tilde U}_L \hat{\tilde U}_{L-1}\dots \hat{\tilde U}_1$ in the RCVQC model and those in the CCVQC model as $\hat{U}= \hat\Lambda_d \hat{\tilde U}_{L} \hat\Lambda_d \hat{\tilde U}_{L-1} \hat\Lambda_d \dots \hat\Lambda_d  \hat{\tilde U}_1\hat\Lambda_d $. We can use the triangle inequality to express the difference between the states after evolutions of the initial states $\hat{\rho}$ in each model,
\begin{align}
    \label{eq:triangle-inequality}
    \|\hat{\tilde U}\hat \rho\hat{\tilde U}^\dagger-\hat U\hat\rho\hat U^\dagger\|_1\leq \sum_{j=0}^{L} \| \hat \rho_j - \frac{1}{\mathcal N_j}\hat \Lambda_d \hat\rho_j \hat\Lambda_d\|_1
\end{align}
where
\begin{align}
    \hat \rho_j= \hat\Lambda_d\hat{\tilde U}_j\hat\Lambda_d\hat{\tilde U}_{j-1}\dots \hat\Lambda_d \hat{\tilde U}_{1} \hat\Lambda_d \hat\rho \hat\Lambda_d \hat{\tilde U}_1^\dagger \hat\Lambda_d \dots \hat{\tilde U}_{j-1}^\dagger\hat\Lambda_d \hat{\tilde U}_j^\dagger \hat\Lambda_d.
\end{align}

We note that the probability of measuring the state with this operator is given by
\begin{align}
    \Tr(\hat \Lambda_{d} \hat\rho )=&\int_\zeta \, \text{d} \mathbf x \bra{\hat q=\mathbf x}\hat \rho \ket{\hat q=\mathbf x}\nonumber\\
    =&\int_{\zeta} \, \text{d} \mathbf x \text{Pr}(\hat q_j=\mathbf x)\nonumber\\
    \geq&1-\sum_{j=0}^{n-1}\text{Pr} \left(|\hat q_j| \geq \frac{d-1}{2}\ell \right)
\end{align}
where $\zeta=\left[-\left(\lfloor d/2\rfloor+1/2 \right)\ell, (\lceil d/2\rceil-1/2)\ell\right)^{\times n}$.
We can then use Markov's inequality to find
\begin{align}
\text{Pr} \left(|\hat q_j| \geq \frac{d-1}{2}\ell \right)&=\text{Pr}\left(\hat q_j^{2}\geq \left(\frac{d-1}{2}\ell \right)^{2}\right)\nonumber\\
&\leq\frac{  \left\langle  \hat{q}_j^{2} \right\rangle_{\hat\rho}}{\left(\frac{d-1}{2}\ell \right)^{2}},
\end{align}
and therefore,
\begin{align}
    \Tr(\hat \Lambda_{d} \hat\rho )\geq 1-\frac{\sum_j \langle\hat q_j^2\rangle_{\hat\rho}}{\left(\frac{d-1}{2}\ell \right)^{2}} \geq 1-\frac{4dE_{\hat\rho} }{(d-1)^2 \pi}
\end{align}
since $\hat n=\hat a^\dagger \hat a=\frac{1}{2}(\hat q-i\hat p)(\hat q+i\hat p)=\frac{1}{2}(\hat q^2+\hat p^2-1)$, which means we have $\sum_j\left\langle \hat q_j^2\right\rangle_{\hat\rho} \leq  \sum_j(2\langle \hat n_j\rangle_{\hat\rho}+1)=2E_{\hat \rho}$.
Note that since $d \geq 2$ we get
\begin{align}
    \label{eq:bound-q}
    \Tr(\hat \Lambda_{d} \hat\rho ) \geq 1-\frac{4dE_{\hat\rho} }{(d-d/2)^2 \pi}\geq 1-\frac{16E_{\hat\rho}}{d\pi}.
\end{align}

Now we use Lemma~\ref{lemma:Gentle} which says that for a given $\epsilon\geq 1-\Tr(\hat \Lambda_d \hat\rho_j)$, we have
\begin{align}
    \frac 1 2\left\| \hat \rho_j - \frac{1}{\mathcal N_j}\hat \Lambda_d \hat\rho_j \hat\Lambda_d\right\|_1\leq \sqrt{\epsilon}\leq 4\sqrt{\frac{E_{\hat \rho_j}}{d\pi}},
\end{align}
and therefore, by Eq.~(\ref{eq:triangle-inequality}), we have
\begin{align}
    \frac 1 2\left\|\hat{\tilde U}\hat \rho\hat{\tilde U}^\dagger-\hat U\hat\rho\hat U^\dagger\right\|_1\leq 4(L+1)\sqrt{\frac{E^*}{d\pi}}.
\end{align}
Hence, from Eq.~(\ref{eq:tv-trace}) we get the following bound on the TVD
\begin{align}
    \epsilon_{RC}\leq 4(L+1)\sqrt{\frac{E^*}{d\pi}}.
\end{align}
\end{proof}
Considering that the only difference between CCVQC and MCVQC lies in the definition of the measurement, we only need to consider the output PDFs to compare the models.
\begin{lemma}
    \label{lemma:MCVQC-CCVQC}
    MCVQC is equivalent to CCVQC.
\end{lemma}
\begin{proof}
Note that the state prior to the measurement in both models is described by $\hat \rho'=\hat{U}\hat\rho \hat{U}^\dagger$ and therefore $\hat\Lambda_d \hat\rho'\hat\Lambda_d=\hat\rho'$. The difference between MCVQC and CCVQC lies in the measurement projection operator. For CCVQC, it is defined in terms of Eq.~(\ref{eq:kraus-main}). We see that
    \begin{align}
        \Tr(\hat \rho' \hat K_{\bar{\mathbf x}})=&\Tr(\hat\Lambda_d\hat \rho'\hat\Lambda_d \hat K_{\bar{\mathbf x}}).
    \end{align}
    Finally, by using that $\hat\Lambda_d \hat K_{\bar{\mathbf x}}\hat\Lambda_d=\hat {\tilde K}_{\bar{\mathbf x}}$, as shown explicitly in Appendix~\ref{appendix:restriction-mod}, we have $p_M(\bar{\mathbf x})=p_C(\bar{\mathbf x})$.
\end{proof}
Next, we investigate how well DVQC approximates MCVQC by first comparing the PDFs of the models using SSD and then bounding the error of approximation.
\begin{lemma}
    \label{lemma:MCVQC-SSD}
    MCVQC is equivalent to a model whereby we first perform the SSD on the final CV state before measuring the resulting DV qudit.
\end{lemma}
\begin{proof}
    Note that the PDF of the measurement outcomes of MCVQC is given by
    \begin{align}
        \label{eq:MCVQC-PDF}
        &p_M(\bar{\mathbf x})\nonumber\\
        =& \int_{\mathbb T_{\ell}^n} \dd \mathbf s \sum_{\mathbf m \in \mathbb Z^{n}}\bra{\hat{\mathbf q}=\mathbf{\bar x}+\ell d \mathbf m+\mathbf s} \hat \rho \ket{\hat{\mathbf q}=\mathbf{\bar x}+\ell d \mathbf m+\mathbf s}.
    \end{align}
    We will now see that this is equivalent to performing the SSD and then measuring qudits in the computational basis (where each basis state is shifted by $\lfloor d/2\rfloor$) with outcome $\mathbf a$ corresponding to $\bar{\mathbf x}=\ell \mathbf a$, where $a_j=\{a'_j\}_d$ with $a'_j\in\mathbb Z_d$.

    SSD on a state $\hat \rho$ is given by Eq.~(\ref{eq:ssd}). Taking the state with measurement values $\mathbf a$, we find
    \begin{align}
        \bra{\mathbf a} \Tr_S(\hat \rho) \ket{\mathbf a}=&\int_{\mathbb T_{\ell}^n} \text{d} \mathbf{t}\, \bra{\mathbf a_{\text{GKP}}}\hat D(\mathbf t)\hat\rho  \hat D^\dagger(\mathbf t)\ket{\mathbf a_{\text{GKP}}}.
    \end{align}
    Noting that $\ket{\mathbf a_{\text{GKP}}}=\sum_{\mathbf m}\ket{(\mathbf a+d\mathbf m)\ell}$ we see that this is equal to
    \begin{align}
        &\bra{\mathbf a} \Tr_S(\hat \rho) \ket{\mathbf a}\nonumber\\
        =&\int_{\mathbb T_{\ell}^n} \text{d} \mathbf{t}\, \sum_{\mathbf m,\mathbf m'}\bra{(\mathbf a+d\mathbf m)\ell}\hat D(\mathbf t)\hat\rho  \hat D^\dagger(\mathbf t)\ket{(\mathbf a+d\mathbf m')\ell}\nonumber\\
        =&\int_{\mathbb T_{\ell}^n} \text{d} \mathbf{t}\, \sum_{\mathbf m,\mathbf m'}e^{i\mathbf{t_q}\cdot (\mathbf m'-\mathbf m)d\ell}\nonumber\\
        &\times \bra{(\mathbf a+d\mathbf m)\ell+\mathbf{t_p}}\hat\rho  \ket{(\mathbf a+d\mathbf m')\ell+\mathbf{t_p}}.
    \end{align}
    We then use that
    \begin{align}
        \int_{-\ell/2}^{\ell/2} \dd t e^{itm d\ell}=\ell \operatorname{sinc}(\tfrac 1 2 d\ell^2 m)=\ell \operatorname{sinc}(\pi m)
    \end{align}
    and that $\operatorname{sinc}(\pi m)=0$ for all integers except zero, i.e., $m\in\mathbb Z \setminus \{0\}$, for which case $\operatorname{sinc}(0)=1$. This means, for integer $m$ we can write $\operatorname{sinc}(\pi m)=\delta_{m,0}$. We therefore find that
    \begin{align}
        &\bra{\mathbf a} \Tr_S(\hat \rho) \ket{\mathbf a}\nonumber\\
        =&\int_{\mathbb T_{\ell}^n} \text{d} \mathbf{t_p}\, \sum_{\mathbf m}\bra{(\mathbf a+d\mathbf m)\ell+\mathbf{t_p}}\hat\rho  \ket{(\mathbf a+d\mathbf m)\ell+\mathbf{t_p}},
    \end{align}    
    which is precisely the same as Eq.~(\ref{eq:MCVQC-PDF}), up to a factor of $\ell$ on the input integers.
\end{proof}
\begin{lemma}
    \label{lemma:DVQC-MCVQC}
    DVQC approximates MCVQC, with gates selected from the set described in Fig.~\ref{fig:Universal_circuit_a}, with $n$ modes and up to $K$ rounds of interlaced Gaussian and cubic phase gates with an error bounded by $\epsilon_{MD} \leq 1215 E^{*2} \frac{Kn^2}{\sqrt d}$, where the energy of the CV system is bound by $E^*$ throughout the evolution of the circuit and $d$ is the dimension of the simulating qudit system.
\end{lemma}
\begin{proof}
We first note that the probability density function $p_M(\bar{\mathbf x})$ for the model MCVQC  can be expressed using Lemma~\ref{lemma:MCVQC-SSD} as
\begin{align}
    p_M(\bar{\mathbf x}=\ell \mathbf a)=\bra{\mathbf a} \Tr_S(\hat U \hat \rho_0 \hat U^\dagger) \ket{\mathbf a}.
\end{align}
We will now see that this model can be simulated by initializing a DV system in the state $\Tr_S(\hat\rho_0)$, followed by implementing DV gates $\hat U_d$, each corresponding to relevant CV gates $\hat U$. We express the DV PDF as
\begin{align}
    p_D(\mathbf a)=\bra{\mathbf a} \hat U_d\Tr_S( \hat \rho_0) \hat U_d^\dagger \ket{\mathbf a}.
\end{align}
From Eq.~(\ref{eq:tv-trace}) we know that we can bound the TVD between these two probability functions in terms of the trace distance. By also using the triangle inequality, we get
\begin{align}
    \delta_{\text{TV}}(p_M,p_D)&\leq\frac 1 2 \|\Tr_S(\hat U\hat \rho \hat U^\dagger)-\hat U_d\Tr_S(\hat \rho) \hat U_d^\dagger\|_1\nonumber\\
    &\leq \frac 1 2\sum_j \|\Tr_S(\hat U_j\hat \rho_j \hat U_j^\dagger)- \hat U_{d,j} \Tr_S(\hat \rho_j) \hat U_{d,j}^\dagger \|_1,
\end{align}
where here $\hat U_{j}$ and $\hat U_{d,j}$ refer to the $j$-th elementary gate in the sequence of the CV gates and DV gates, respectively.
    We consider the difference between initiating with a DV state derived from a CV state by the SSD, acting with DV operations and measurements, compared to initiating with a CV state, acting with CV operations, and then performing the SSD before measuring with DV measurements.
    
    To evaluate this error, we first note that we apply $K$ rounds of passive Gaussian and displacements interlaced with $K$ cubic phase gates. We also apply a final round of two passive operations and a round of squeezing. In total, we have $K+2$ passive operations with displacements. We also have $K$ cubic phase gates and $K$ squeezing operations. For each passive gate, $n(n-1)/2$ Mach-Zehnder interferometers are applied along with $n$ single-mode rotations. In addition, each round consists of $n$ displacements.
    
     Using Lemma~\ref{lemma:error-beamsplitter}, we see that the $(K+2)n(n-1)/2$ Mach-Zehnder interferometer operations increase the trace distance by a maximum of
    \begin{align}
        \label{eq:error-bs}
        108 (K+2)n(n-1) \frac{\sqrt{E^*}}{d}\leq 324 Kn^2\frac{\sqrt{E^*}}{d}.
    \end{align}
    Using Lemma~\ref{lemma:error-rotation}, the $(K+2)n$ single-mode rotations contribute a trace distance bounded by
    \begin{align}
        52(K+2)n\frac{\sqrt{E^*}}{d}\leq 156Kn\frac{\sqrt{E^*}}{d}.
    \end{align}
    
    We also see that at most $(K+2)n$ displacements are applied, which does not affect the error.
    Using Lemma~\ref{lemma:error-cubic-phase}, we see that a maximum of $K$ cubic phase gates increases the trace distance by a maximum of
    \begin{align}
        223 KE^{*2}/\sqrt{d}.
    \end{align}
    Finally, using Lemma~\ref{lemma:error-squeezing}, we see that a maximum of $n$ squeezing operations corresponds to a maximum increase in trace distance of
    \begin{align}
        32 n \frac{E^{*3/2}}{d}.
    \end{align}
    Putting it together, we obtain a bound on the total trace distance given by
    \begin{align}
        \left(324 Kn^2+156Kn+223KE^{*3/2}\sqrt{d}+32 E^{*}n\right) \frac{\sqrt{E^*}}{d}.
    \end{align}
    This can be simplified by using $d \geq 2$, $K,n \geq 1$ and $E^{*} \geq \frac 1 2$, hence
    \begin{align}
        \epsilon_{MD} \leq& 1215 E^{*2} \frac{Kn^2}{\sqrt d}.
    \end{align}
\end{proof}

\section{Equivalence between realistic continuous- and discrete-variable quantum computing}
\label{sec:equivalence}
In this Section, we provide a proof of the main Theorem of this work, along with the corresponding Corollary, demonstrating simulatability with qubits.

\begin{proof}
    (Proof of Theorem \ref{theorem:main-result})  Here we use the triangle inequality to connect each of the models according to Fig.~\ref{fig:Comparison_summary}. We consider the total number of elementary gates that are applied as being bounded by
    \begin{align}
        L&= (K+2)(n(n-1)/2)+2(K+2)n+K+n\nonumber\\
        &\leq 3Kn^2/2+6Kn+K+n\nonumber\\
        &\leq 10Kn^2,
    \end{align}
    and therefore $L+1\leq 11 Kn^2$.
    The total error, when considering the contribution from Lemma~\ref{lemma:RCVQC-CCVQC}, Lemma~\ref{lemma:MCVQC-CCVQC} and Lemma~\ref{lemma:DVQC-MCVQC}, is given by a quantity that is upper-bounded by
    \begin{align}
        \epsilon \leq & \epsilon_{MD}+\epsilon_{RC}\nonumber\\
        \leq & 1215 E^{*2} \frac{Kn^2}{\sqrt d} + 44Kn^2 \sqrt{\frac{E^*}{d\pi}}\nonumber\\
        \leq & \left( 1215 + \frac{44}{\sqrt{\pi/8}} \right) E^{*2} \frac{Kn^2}{\sqrt d} \nonumber\\
        \leq & 1286 E^{*2} \frac{Kn^2}{\sqrt d},
    \end{align}
    where we again used the fact that $E^{*} \geq \frac 1 2$.
\end{proof}

This Theorem demonstrates that with a total energy bound of $E^*$, $n$ modes and $K$ cubic phase gates, the circuit can be simulated using $n$ qudits of dimension $d$ up to an arbitrary error $\epsilon$ whereby
\begin{align}
    d\geq (1286 Kn^2E^{*2}/\epsilon)^2.
\end{align}

We can also encode qudits as qubits, enabling us to simulate CV circuits using qubits. We now provide a proof that each of the $d$-dimensional qudit operations can be implemented on a qubit device, where $d=2^k$ and $k$ is the number of qubits needed to represent each qudit. Note that in the following analysis, we only consider the space complexity of simulating the CV circuits with qubits.
\begin{proof}
(Proof of Corollary \ref{corollary:qubits}.)
We interpret the $a$-th level of the qudit as the representation of $a$ in binary form, encoded into qubit basis states. I.e., given an integer $a$, we can write its binary representation as $a_1a_2\dots a_k$, and the corresponding qubit representation of each qudit basis state $\ket a$ is $\ket{a_1}\ket{a_2}\dots\ket{a_n}$.

First, assuming we begin with the vacuum state $\ket\emptyset$, we can evaluate the input qudit as $\Tr_S(\ketbra{\emptyset}{\emptyset})$, which can be encoded within $k$ qubits as they have the same Hilbert space.

The qudit Fourier transform, specified by Eq.~\ref{eq: dv Fourier transform}, can be decomposed as the quantum Fourier transform across $k$ qubits~\cite{nielsen2010}, which can be implemented with $k(k+1)/2$ qubit gates.

The single-qudit $\hat Z_d(s)$ operations correspond to applying a diagonal unitary operation across $k$ qubits, which requires $2^k-1$ qubit gates.
Furthermore, for the Mach-Zehnder interferometer, we use the decompositions given in Eq.~(\ref{eq:Mach-Zehnder}) and Eq.~(\ref{eq:beam-splitter}) to decompose it in terms of rotations and controlled-$Z$ operations. The two-mode controlled-$Z$ gates are diagonal in the basis of the $2k$ qubits representing two modes and therefore require $2^{2k}-1$ qubit gates to implement them. Finally, rotations can be decomposed in terms of shear gates and Fourier transforms. The shear gate is also diagonal in the basis of $k$ qubits and can be decomposed as a sequence of $2^k-1$ qubit gates. This decomposition from $n$ qudits of dimension $2^k$ to $kn$ qubits is exact.
\end{proof}

Note that, by simulating a CV system with $pn$ qubits, we can produce outcomes which have an error of
\begin{align}
    \epsilon \leq 1286 E^{*2} Kn^2 2^{-p/2},
\end{align}
which implies that we can simulate $n$ modes up to arbitrary error $\epsilon$ using at least
\begin{align}
    \label{eq:kmin}
    k  \geq 2\log_2 (1286 Kn^2E^{2*}/\epsilon)
\end{align}
qubits per mode. Note that this demonstrates that the space complexity scales logarithmically with respect to the number of modes for constant energy. It also implies that the scaling is polynomial for exponentially increasing energy.

However, note that the number of qubit gates required to simulate each CV gate increases exponentially with respect to the number of qubits per mode. Therefore, the total number of qubit gates required to simulate a CV circuit scales polynomially with respect to the energy.

\section{Conclusion}
\label{sec:conclusion}
In conclusion, we provide a framework for mapping a gate-based CVQC model consisting of states and operations with finite energy onto a DVQC device. At the same time, we also show that there cannot be an exponential advantage in the computational power of CV devices with finite energy compared to devices based on the DV paradigm. This has been a long-standing question without any rigorous answer.

We note that practical use cases may be limited by the high number of qubits required to simulate each mode. However, we believe that the bounds we have derived for the error are relatively weak. We expect that tighter bounds are possible to derive, which we leave for future investigations.

For other future work, we envision the construction of a method to simulate CV circuits with classical means using existing DV simulation methods~\cite{vidal2003,aaronson2004,van-den-nest2010,ohliger2012,mari2012,veitch2012,pashayan2015,bravyi2016,bennink2017, bu2019, seddon2021, Marshall:23,  hahn2024b,PhysRevA.99.062337}. This result will deepen existing knowledge about the relationship between CVQC and DVQC. 

Furthermore, we highlight the possibility of using our framework to formulate a rigorous method to convert CV-algorithms (particularly those involving GKP-encoded states), such as finite-energy versions of Refs.~\cite{anschuetz2022a,Anschuetz_2026,Brenner_2025}, to DV-algorithms. 

While the present analysis was limited to the specific choice of operations generated by Gaussian gates and the cubic phase gate, we would expect the same analysis to extend to different choices of elementary gates, including energy-well-behaved gates such as the Kerr gate.
We leave the investigation of whether it is possible to simulate arbitrary polynomial gates using our method---akin to that presented in Ref.~\cite{arzani2025}---to future work.

\paragraph*{Note added.} During the completion of this work, we became aware of a related work by Ulysse Chabaud et al. that was submitted to the arXiv simultaneously~\cite{chabaud2025b}.

\begin{acknowledgments}We thank Ulysse Chabaud for useful discussions and for providing valuable comments on an initial draft of our manuscript.
We also acknowledge useful discussions with Francesco Arzani, Dan Browne, Alessandro Ferraro, Xun Gao, David Gross, Oliver Hahn, Timo Hillmann, P\'erola Milman, Jonas S. Neergaard-Nielsen, Olga Solodovnikova, and Rui Wang. G.F.\ acknowledges funding from the European Union’s Horizon Europe Framework Programme (EIC Pathfinder Challenge project Veriqub) under Grant Agreement No.\ 101114899 and  the Olle Engkvist Foundation.
G.F. and A.M. acknowledge financial support from the Swedish Research Council through the project grant VR DAIQUIRI.  G.F., C.C. and L.R. acknowledge support from the Knut and Alice Wallenberg Foundation through the Wallenberg Center for Quantum Technology (WACQT).
\end{acknowledgments}

\bibliographystyle{apsrev4-2}
\bibliography{./main.bib}

@article{aaronson2004,
 author = {Aaronson, Scott and Gottesman, Daniel},
 doi = {10.1103/PhysRevA.70.052328},
 issue = {5},
 journal = {Phys. Rev. A},
 month = {Nov},
 numpages = {14},
 pages = {052328},
 publisher = {American Physical Society},
 title = {Improved simulation of stabilizer circuits},
 url = {https://link.aps.org/doi/10.1103/PhysRevA.70.052328},
 volume = {70},
 year = {2004}
}

@article{anschuetz2022a,
 author = {Anschuetz, Eric R. and Hu, Hong-Ye and Huang, Jin-Long and Gao, Xun},
 doi = {10.1103/PRXQuantum.4.020338},
 issue = {2},
 journal = {PRX Quantum},
 month = {Jun},
 numpages = {22},
 pages = {020338},
 publisher = {American Physical Society},
 title = {Interpretable Quantum Advantage in Neural Sequence Learning},
 url = {https://link.aps.org/doi/10.1103/PRXQuantum.4.020338},
 volume = {4},
 year = {2023}
}

@article{Anschuetz_2026,
 author = {Anschuetz, Eric R. and Gao, Xun},
 doi = {10.22331/q-2026-01-20-1976},
 issn = {2521-327X},
 journal = {Quantum},
 month = {January},
 pages = {1976},
 publisher = {Verein zur Forderung des Open Access Publizierens in den Quantenwissenschaften},
 title = {Arbitrary Polynomial Separations in Trainable Quantum Machine Learning},
 url = {http://dx.doi.org/10.22331/q-2026-01-20-1976},
 volume = {10},
 year = {2026}
}

@article{arvind1995,
 author = {{Arvind} and Dutta, B and Mukunda, N and Simon, R},
 doi = {10.1007/BF02848172},
 journal = {Pramana J. Phys.},
 langid = {english},
 pages = {441--497},
 title = {The Real Symplectic Groups in Quantum Mechanics and Optics},
 volume = {45},
 year = {1995}
}

@article{arzani2025,
 author = {Arzani, Francesco and Booth, Robert I and Chabaud, Ulysse},
 doi = {10.1038/s41467-025-64872-3},
 journal = {Nature Communications},
 number = {1},
 pages = {9744},
 publisher = {Nature Publishing Group UK London},
 title = {Effective descriptions of bosonic systems can be considered complete},
 volume = {16},
 year = {2025}
}

@article{banuls2020,
 author = {Bañuls, Mari Carmen and Blatt, Rainer and Catani, Jacopo and Celi, Alessio and Cirac, Juan Ignacio and Dalmonte, Marcello and Fallani, Leonardo and Jansen, Karl and Lewenstein, Maciej and Montangero, Simone and Muschik, Christine A. and Reznik, Benni and Rico, Enrique and Tagliacozzo, Luca and Van Acoleyen, Karel and Verstraete, Frank and Wiese, Uwe-Jens and Wingate, Matthew and Zakrzewski, Jakub and Zoller, Peter},
 doi = {10.1140/epjd/e2020-100571-8},
 issn = {1434-6079},
 journal = {The European Physical Journal D},
 number = {8},
 pages = {165},
 title = {Simulating Lattice Gauge Theories within Quantum Technologies},
 url = {https://doi.org/10.1140/epjd/e2020-100571-8},
 volume = {74},
 year = {2020}
}

@article{bennink2017,
 author = {Bennink, Ryan S. and Ferragut, Erik M. and Humble, Travis S. and Laska, Jason A. and Nutaro, James J. and Pleszkoch, Mark G. and Pooser, Raphael C.},
 doi = {10.1103/PhysRevA.95.062337},
 issue = {6},
 journal = {Phys. Rev. A},
 month = {Jun},
 numpages = {10},
 pages = {062337},
 publisher = {American Physical Society},
 title = {Unbiased simulation of near-Clifford quantum circuits},
 url = {https://link.aps.org/doi/10.1103/PhysRevA.95.062337},
 volume = {95},
 year = {2017}
}

@misc{blair2025faulttolerantquantumcomputationuniversal,
 archiveprefix = {arXiv},
 author = {Sheron Blair and Francesco Arzani and Giulia Ferrini and Alessandro Ferraro},
 eprint = {2506.13643},
 primaryclass = {quant-ph},
 title = {Towards fault-tolerant quantum computation with universal continuous-variable gates},
 url = {https://arxiv.org/abs/2506.13643},
 year = {2025}
}

@article{bourassa2021fast,
 author = {Bourassa, J. Eli and Quesada, Nicol{\'a}s and Tzitrin, Ilan and Sz{\'a}va, Antal and Isacsson, Theodor and Izaac, Josh and Sabapathy, Krishna Kumar and Dauphinais, Guillaume and Dhand, Ish},
 doi = {10.1103/PRXQuantum.2.040315},
 journal = {PRX Quantum},
 month = {October},
 number = {4},
 pages = {040315},
 publisher = {{American Physical Society}},
 title = {Fast Simulation of Bosonic Qubits via Gaussian Functions in Phase Space},
 volume = {2},
 year = {2021}
}

@article{Brady_2024,
 author = {Brady, Anthony J. and Eickbusch, Alec and Singh, Shraddha and Wu, Jing and Zhuang, Quntao},
 doi = {10.1016/j.pquantelec.2023.100496},
 issn = {0079-6727},
 journal = {Prog. Quantum Electron.},
 month = {January},
 pages = {100496},
 publisher = {Elsevier BV},
 title = {Advances in bosonic quantum error correction with Gottesman–Kitaev–Preskill Codes: Theory, engineering and applications},
 url = {http://dx.doi.org/10.1016/j.pquantelec.2023.100496},
 volume = {93},
 year = {2024}
}

@article{braunstein2005,
 author = {Braunstein, Samuel L. and van Loock, Peter},
 doi = {10.1103/RevModPhys.77.513},
 issue = {2},
 journal = {Rev. Mod. Phys.},
 month = {Jun},
 numpages = {0},
 pages = {513--577},
 publisher = {American Physical Society},
 title = {Quantum information with continuous variables},
 url = {https://link.aps.org/doi/10.1103/RevModPhys.77.513},
 volume = {77},
 year = {2005}
}

@article{bravyi2016,
 author = {Bravyi, Sergey and Gosset, David},
 doi = {10.1103/PhysRevLett.116.250501},
 issue = {25},
 journal = {Phys. Rev. Lett.},
 month = {Jun},
 numpages = {5},
 pages = {250501},
 publisher = {American Physical Society},
 title = {Improved Classical Simulation of Quantum Circuits Dominated by Clifford Gates},
 url = {https://link.aps.org/doi/10.1103/PhysRevLett.116.250501},
 volume = {116},
 year = {2016}
}

@article{Brenner_2025,
 author = {Brenner, Lukas and Caha, Libor and Coiteux-Roy, Xavier and Koenig, Robert},
 doi = {10.1038/s41467-025-67694-5},
 journal = {Nature Communications},
 pages = {227},
 title = {Factoring an integer with three oscillators and a qubit},
 volume = {17},
 year = {2025}
}

@article{bu2019,
 author = {Bu, Kaifeng and Koh, Dax Enshan},
 doi = {10.1103/PhysRevLett.123.170502},
 issue = {17},
 journal = {Phys. Rev. Lett.},
 month = {Oct},
 numpages = {5},
 pages = {170502},
 publisher = {American Physical Society},
 title = {Efficient Classical Simulation of Clifford Circuits with Nonstabilizer Input States},
 url = {https://link.aps.org/doi/10.1103/PhysRevLett.123.170502},
 volume = {123},
 year = {2019}
}

@article{calcluth2022,
 author = {Calcluth, Cameron and Ferraro, Alessandro and Ferrini, Giulia},
 doi = {10.22331/q-2022-12-01-867},
 issn = {2521-327X},
 journal = {Quantum},
 month = {December},
 pages = {867},
 publisher = {{Verein zur F\"orderung des Open Access Publizierens in den Quantenwissenschaften}},
 title = {Efficient Simulation of {{Gottesman-Kitaev-Preskill}} States with {{Gaussian}} Circuits},
 volume = {6},
 year = {2022}
}

@article{calcluth2023,
 author = {Calcluth, Cameron and Ferraro, Alessandro and Ferrini, Giulia},
 doi = {10.1103/PhysRevA.107.062414},
 journal = {Phys. Rev. A},
 month = {June},
 number = {6},
 pages = {062414},
 publisher = {{American Physical Society}},
 title = {Vacuum Provides Quantum Advantage to Otherwise Simulatable Architectures},
 volume = {107},
 year = {2023}
}

@article{Calcluth2024,
 author = {Calcluth, Cameron and Reichel, Nicolas and Ferraro, Alessandro and Ferrini, Giulia},
 doi = {10.1103/PRXQuantum.5.020337},
 issue = {2},
 journal = {PRX Quantum},
 month = {May},
 numpages = {32},
 pages = {020337},
 publisher = {American Physical Society},
 title = {Sufficient Condition for Universal Quantum Computation Using Bosonic Circuits},
 url = {https://link.aps.org/doi/10.1103/PRXQuantum.5.020337},
 volume = {5},
 year = {2024}
}

@article{Calcluth2025,
 author = {Calcluth, Cameron and Hahn, Oliver and Bermejo-Vega, Juani and Ferraro, Alessandro and Ferrini, Giulia},
 doi = {10.1103/xmtw-g54f},
 issue = {1},
 journal = {Phys. Rev. Lett.},
 month = {Jul},
 numpages = {7},
 pages = {010601},
 publisher = {American Physical Society},
 title = {Classical Simulation of Circuits with Realistic Odd-Dimensional Gottesman-Kitaev-Preskill States},
 url = {https://link.aps.org/doi/10.1103/xmtw-g54f},
 volume = {135},
 year = {2025}
}

@misc{chabaud2025b,
 archiveprefix = {arXiv},
 author = {Ulysse Chabaud and Sevag Gharibian and Saeed Mehraban and Arsalan Motamedi and Hamid Reza Naeij and Dorian Rudolph and Dhruva Sambrani},
 eprint = {2510.08545},
 primaryclass = {quant-ph},
 title = {Energy, Bosons and Computational Complexity},
 url = {https://arxiv.org/abs/2510.08545},
 year = {2025}
}

@misc{chabaudprivate,
 author = {Chabaud, Ulysse},
 date = {2025-04-18},
 howpublished = {Private communication},
 year = {2025}
}

@article{Childs2022quantumsimulationof,
 author = {Childs, Andrew M. and Leng, Jiaqi and Li, Tongyang and Liu, Jin-Peng and Zhang, Chenyi},
 doi = {10.22331/q-2022-11-17-860},
 issn = {2521-327X},
 journal = {{Quantum}},
 month = {November},
 pages = {860},
 publisher = {{Verein zur F{\"{o}}rderung des Open Access Publizierens in den Quantenwissenschaften}},
 title = {Quantum simulation of real-space dynamics},
 url = {https://doi.org/10.22331/q-2022-11-17-860},
 volume = {6},
 year = {2022}
}

@article{clements2016,
 author = {Clements, William R. and Humphreys, Peter C. and Metcalf, Benjamin J. and Kolthammer, W. Steven and Walsmley, Ian A.},
 doi = {10.1364/OPTICA.3.001460},
 issn = {2334-2536},
 journal = {Optica},
 langid = {english},
 month = {December},
 number = {12},
 pages = {1460},
 title = {Optimal Design for Universal Multiport Interferometers},
 urldate = {2021-08-09},
 volume = {3},
 year = {2016}
}

@article{cochrane1999,
 author = {Cochrane, Paul T and Milburn, Gerard J and Munro, William J},
 doi = {10.1103/PhysRevA.59.2631},
 journal = {Phys. Rev. A},
 number = {4},
 pages = {2631},
 publisher = {{APS}},
 title = {Macroscopically Distinct Quantum-Superposition States as a Bosonic Code for Amplitude Damping},
 volume = {59},
 year = {1999}
}

@article{descamps2024,
 author = {Descamps, Eloi and Fabre, Nicolas and Saharyan, Astghik and Keller, Arne and Milman, Pérola},
 date = {2024-12-27},
 doi = {10.1103/PhysRevLett.133.260605},
 issn = {0031-9007, 1079-7114},
 journal = {Phys. Rev. Lett.},
 number = {26},
 pages = {260605},
 shortjournal = {Phys. Rev. Lett.},
 title = {Superselection {{Rules}} and {{Bosonic Quantum Computational Resources}}},
 url = {https://link.aps.org/doi/10.1103/PhysRevLett.133.260605},
 urldate = {2025-06-12},
 volume = {133},
 year = {2024}
}

@article{Descamps_2026,
 author = {Descamps, Eloi and Saharyan, Astghik and Chivet, Adrien and Keller, Arne and Milman, Pérola},
 doi = {10.1364/opticaq.581218},
 issn = {2837-6714},
 journal = {Optica Quantum},
 month = {April},
 number = {2},
 pages = {148},
 publisher = {Optica Publishing Group},
 title = {Unified framework for bosonic quantum information encoding, resources, and universality from superselection rules},
 url = {http://dx.doi.org/10.1364/OPTICAQ.581218},
 volume = {4},
 year = {2026}
}

@article{dias2024classical,
 author = {Dias, Beatriz and K{\"o}nig, Robert},
 doi = {https://doi.org/10.1103/PhysRevA.110.042402},
 journal = {Phys. Rev. A},
 number = {4},
 pages = {042402},
 publisher = {APS},
 title = {Classical simulation of non-Gaussian bosonic circuits},
 volume = {110},
 year = {2024}
}

@article{dolag2008,
 author = {Dolag, K. and Borgani, S. and Schindler, S. and Diaferio, A. and Bykov, A. M.},
 date = {2008-02-01},
 doi = {10.1007/s11214-008-9316-5},
 issn = {1572-9672},
 journal = {Space Science Reviews},
 number = {1},
 pages = {229--268},
 title = {Simulation {{Techniques}} for {{Cosmological Simulations}}},
 url = {https://doi.org/10.1007/s11214-008-9316-5},
 volume = {134},
 year = {2008}
}

@article{douce2019,
 author = {Douce, Tom and Markham, Damian and Kashefi, Elham and van Loock, Peter and Ferrini, Giulia},
 doi = {10.1103/PhysRevA.99.012344},
 journal = {Phys. Rev. A},
 month = {January},
 number = {1},
 pages = {012344},
 title = {Probabilistic fault-tolerant universal quantum computation and sampling problems in continuous variables},
 url = {https://link.aps.org/doi/10.1103/PhysRevA.99.012344},
 urldate = {2019-06-11},
 volume = {99},
 year = {2019}
}

@book{ferraro2005,
 address = {Napoli},
 archiveprefix = {arXiv},
 author = {A. Ferraro and S. Olivares and M. G. A. Paris},
 eprint = {quant-ph/0503237},
 publisher = {Bibliopolis},
 title = {Gaussian States in Quantum Information},
 year = {2005}
}

@inproceedings{garcia-alvarez2019,
 address = {{Singapore}},
 author = {{Garc{\'i}a-{\'A}lvarez}, L. and Ferraro, A. and Ferrini, G.},
 booktitle = {International Symposium on Mathematics, Quantum Theory, and Cryptography},
 doi = {10.1007/978-981-15-5191-8_9},
 editor = {Takagi, Tsuyoshi and Wakayama, Masato and Tanaka, Keisuke and Kunihiro, Noboru and Kimoto, Kazufumi and Ikematsu, Yasuhiko},
 isbn = {978-981-15-5191-8},
 pages = {79--92},
 publisher = {{Springer Singapore}},
 title = {From the Bloch Sphere to Phase-Space Representations with the {{Gottesman}}\textendash{{Kitaev}}\textendash{{Preskill}} Encoding},
 year = {2021}
}

@article{gottesman2001,
 author = {Gottesman, Daniel and Kitaev, Alexei and Preskill, John},
 doi = {10.1103/PhysRevA.64.012310},
 journal = {Physical Review A},
 month = {June},
 number = {1},
 pages = {012310},
 publisher = {{American Physical Society}},
 title = {Encoding a Qubit in an Oscillator},
 volume = {64},
 year = {2001}
}

@article{grimsmo2020,
 author = {Grimsmo, Arne L. and Combes, Joshua and Baragiola, Ben Q.},
 doi = {10.1103/PhysRevX.10.011058},
 issue = {1},
 journal = {Phys. Rev. X},
 month = {Mar},
 numpages = {32},
 pages = {011058},
 publisher = {American Physical Society},
 title = {Quantum Computing with Rotation-Symmetric Bosonic Codes},
 url = {https://link.aps.org/doi/10.1103/PhysRevX.10.011058},
 volume = {10},
 year = {2020}
}

@article{hahn2024b,
 author = {Hahn, Oliver and Ferrini, Giulia and Takagi, Ryuji},
 doi = {10.1103/PRXQuantum.6.010330},
 issue = {1},
 journal = {PRX Quantum},
 month = {Feb},
 numpages = {30},
 pages = {010330},
 publisher = {American Physical Society},
 title = {Bridging Magic and Non-Gaussian Resources via Gottesman-Kitaev-Preskill Encoding},
 url = {https://link.aps.org/doi/10.1103/PRXQuantum.6.010330},
 volume = {6},
 year = {2025}
}

@misc{hahn2024c,
 author = {Hahn, Oliver and Takagi, Ryuji and Ferrini, Giulia and Yamasaki, Hayata},
 doi = {10.22331/q-2025-10-13-1881},
 journal = {Quantum},
 pages = {1881},
 publisher = {Verein zur F{\"o}rderung des Open Access Publizierens in den Quantenwissenschaften},
 title = {Classical simulation and quantum resource theory of non-Gaussian optics},
 volume = {9},
 year = {2025}
}

@article{hofstetter2018,
 author = {Hofstetter, W and Qin, T},
 doi = {10.1088/1361-6455/aaa31b},
 issn = {0953-4075},
 journal = {J. Phys. B},
 langid = {english},
 number = {8},
 pages = {082001},
 publisher = {IOP Publishing},
 title = {Quantum Simulation of Strongly Correlated Condensed Matter Systems},
 url = {https://doi.org/10.1088/1361-6455/aaa31b},
 urldate = {2026-02-20},
 volume = {51},
 year = {2018}
}

@article{jordan2012,
 author = {Stephen P. Jordan  and Keith S. M. Lee  and John Preskill },
 doi = {10.1126/science.1217069},
 journal = {Science},
 number = {6085},
 pages = {1130-1133},
 title = {Quantum Algorithms for Quantum Field Theories},
 volume = {336},
 year = {2012}
}

@article{Kalajdzievski_2021,
 author = {Kalajdzievski, Timjan and Quesada, Nicolás},
 doi = {10.22331/q-2021-02-08-394},
 issn = {2521-327X},
 journal = {Quantum},
 month = {February},
 pages = {394},
 publisher = {Verein zur Forderung des Open Access Publizierens in den Quantenwissenschaften},
 title = {Exact and approximate continuous-variable gate decompositions},
 url = {http://dx.doi.org/10.22331/q-2021-02-08-394},
 volume = {5},
 year = {2021}
}

@article{Lamata01012018,
 author = {Lucas Lamata and Adrian Parra-Rodriguez and Mikel Sanz and Enrique Solano},
 doi = {10.1080/23746149.2018.1457981},
 journal = {Advances in Physics: X},
 number = {1},
 pages = {1457981},
 publisher = {Taylor \& Francis},
 title = {Digital-analog quantum simulations with superconducting circuits},
 url = { 

https://doi.org/10.1080/23746149.2018.1457981


},
 volume = {3},
 year = {2018}
}

@article{liu2023simulating,
 author = {Liu, Minzhao and Oh, Changhun and Liu, Junyu and Jiang, Liang and Alexeev, Yuri},
 doi = {10.1103/PhysRevA.108.052604},
 issue = {5},
 journal = {Phys. Rev. A},
 month = {Nov},
 numpages = {16},
 pages = {052604},
 publisher = {American Physical Society},
 title = {Simulating lossy Gaussian boson 
sampling with matrix-product operators},
 url = {https://link.aps.org/doi/10.1103/PhysRevA.108.052604},
 volume = {108},
 year = {2023}
}

@article{lloyd1996,
 author = {Lloyd, S.},
 doi = {10.1126/science.273.5278.1073},
 journal = {Science},
 pages = {1073},
 title = {Universal Quantum Simulators},
 volume = {273},
 year = {1996}
}

@article{lloyd1999,
 author = {Lloyd, Seth and Braunstein, Samuel L.},
 doi = {10.1103/PhysRevLett.82.1784},
 journal = {Phys. Rev. Lett.},
 pages = {1784},
 title = {Quantum Computation over Continuous Variables},
 volume = {82},
 year = {1999}
}

@article{mari2012,
 author = {Mari, A. and Eisert, J.},
 doi = {10.1103/PhysRevLett.109.230503},
 issn = {0031-9007, 1079-7114},
 journal = {Phys. Rev. Lett.},
 langid = {english},
 month = {December},
 number = {23},
 pages = {230503},
 title = {Positive {{Wigner Functions Render Classical Simulation}} of {{Quantum Computation Efficient}}},
 urldate = {2022-11-10},
 volume = {109},
 year = {2012}
}

@article{Marshall:23,
 author = {Jeffrey Marshall and Namit Anand},
 doi = {10.1364/OPTICAQ.504311},
 journal = {Optica Quantum},
 month = {Dec},
 number = {2},
 pages = {78--93},
 publisher = {Optica Publishing Group},
 title = {Simulation of quantum optics by coherent state decomposition},
 url = {https://opg.optica.org/opticaq/abstract.cfm?URI=opticaq-1-2-78},
 volume = {1},
 year = {2023}
}

@article{mocz2021,
 author = {Mocz, Philip and Szasz, Aaron},
 date = {2021-03},
 doi = {10.3847/1538-4357/abe6ac},
 issn = {0004-637X},
 journal = {The Astrophysical Journal},
 number = {1},
 pages = {29},
 publisher = {The American Astronomical Society},
 title = {Toward {{Cosmological Simulations}} of {{Dark Matter}} on {{Quantum Computers}}},
 url = {https://doi.org/10.3847/1538-4357/abe6ac},
 urldate = {2026-02-20},
 volume = {910},
 year = {2021}
}

@book{nielsen2010,
 address = {{Cambridge ; New York}},
 author = {Nielsen, Michael A. and Chuang, Isaac L.},
 doi = { https://doi.org/10.1017/CBO9780511976667},
 edition = {10th anniversary ed},
 isbn = {978-1-107-00217-3},
 langid = {english},
 lccn = {QA76.889 .N54 2010},
 publisher = {{Cambridge University Press}},
 title = {Quantum Computation and Quantum Information},
 year = {2010}
}

@article{oh2023tensor,
 author = {Oh, Changhun and Liu, Minzhao and Alexeev, Yuri and Fefferman, Bill and Jiang, Liang},
 date = {2024/09/01},
 doi = {10.1038/s41567-024-02535-8},
 id = {Oh2024},
 isbn = {1745-2481},
 journal = {Nat. Phys.},
 number = {9},
 pages = {1461--1468},
 title = {Classical algorithm for simulating experimental Gaussian boson sampling},
 url = {https://doi.org/10.1038/s41567-024-02535-8},
 volume = {20},
 year = {2024}
}

@article{ohliger2012,
 author = {Ohliger, M. and Eisert, J.},
 doi = {10.1103/PhysRevA.85.062318},
 issue = {6},
 journal = {Phys. Rev. A},
 month = {Jun},
 numpages = {12},
 pages = {062318},
 publisher = {American Physical Society},
 title = {Efficient measurement-based quantum computing with continuous-variable systems},
 url = {https://link.aps.org/doi/10.1103/PhysRevA.85.062318},
 volume = {85},
 year = {2012}
}

@article{opanchuk2012,
 author = {Opanchuk, Bogdan and Polkinghorne, Rodney and Fialko, Oleksandr and Brand, Joachim and Drummond, Peter D.},
 doi = {https://doi.org/10.1002/andp.201300113},
 journal = {Annalen der Physik},
 number = {10-11},
 pages = {866-876},
 title = {Quantum simulations of the early universe},
 url = {https://onlinelibrary.wiley.com/doi/abs/10.1002/andp.201300113},
 volume = {525},
 year = {2013}
}

@article{pantaleoni2020,
 author = {Pantaleoni, Giacomo and Baragiola, Ben Q and Menicucci, Nicolas C},
 doi = {10.1103/PhysRevLett.125.040501},
 journal = {Phys. Rev. Lett.},
 number = {4},
 pages = {040501},
 publisher = {{APS}},
 title = {Modular Bosonic Subsystem Codes},
 volume = {125},
 year = {2020}
}

@article{pantaleoni2021,
 author = {Pantaleoni, Giacomo and Baragiola, Ben Q and Menicucci, Nicolas C},
 doi = {10.1103/PhysRevA.104.012430},
 journal = {Phys. Rev. A},
 number = {1},
 pages = {012430},
 publisher = {{APS}},
 title = {Subsystem Analysis of Continuous-Variable Resource States},
 volume = {104},
 year = {2021}
}

@article{pantaleoni2023,
 author = {Pantaleoni, Giacomo and Baragiola, Ben Q. and Menicucci, Nicolas C.},
 doi = {10.1103/PhysRevA.107.062611},
 issue = {6},
 journal = {Phys. Rev. A},
 month = {Jun},
 numpages = {15},
 pages = {062611},
 title = {Zak transform as a framework for quantum computation with the Gottesman-Kitaev-Preskill code},
 volume = {107},
 year = {2023}
}

@article{pashayan2015,
 author = {Pashayan, Hakop and Wallman, Joel J. and Bartlett, Stephen D.},
 doi = {10.1103/PhysRevLett.115.070501},
 issue = {7},
 journal = {Phys. Rev. Lett.},
 month = {Aug},
 numpages = {5},
 pages = {070501},
 publisher = {American Physical Society},
 title = {Estimating Outcome Probabilities of Quantum Circuits Using Quasiprobabilities},
 url = {https://link.aps.org/doi/10.1103/PhysRevLett.115.070501},
 volume = {115},
 year = {2015}
}

@article{PhysRevA.109.052412,
 author = {Jha, Raghav G. and Ringer, Felix and Siopsis, George and Thompson, Shane},
 doi = {10.1103/PhysRevA.109.052412},
 issue = {5},
 journal = {Phys. Rev. A},
 month = {May},
 numpages = {16},
 pages = {052412},
 publisher = {American Physical Society},
 title = {Continuous-variable quantum computation of the O(3) model in $1+1$ dimensions},
 url = {https://link.aps.org/doi/10.1103/PhysRevA.109.052412},
 volume = {109},
 year = {2024}
}

@article{PhysRevA.92.063825,
 author = {Marshall, Kevin and Pooser, Raphael and Siopsis, George and Weedbrook, Christian},
 doi = {10.1103/PhysRevA.92.063825},
 issue = {6},
 journal = {Phys. Rev. A},
 month = {Dec},
 numpages = {7},
 pages = {063825},
 publisher = {American Physical Society},
 title = {Quantum simulation of quantum field theory using continuous variables},
 url = {https://link.aps.org/doi/10.1103/PhysRevA.92.063825},
 volume = {92},
 year = {2015}
}

@article{PhysRevA.99.062337,
 author = {Rall, Patrick and Liang, Daniel and Cook, Jeremy and Kretschmer, William},
 doi = {10.1103/PhysRevA.99.062337},
 issue = {6},
 journal = {Phys. Rev. A},
 month = {Jun},
 numpages = {10},
 pages = {062337},
 publisher = {American Physical Society},
 title = {Simulation of qubit quantum circuits via Pauli propagation},
 url = {https://link.aps.org/doi/10.1103/PhysRevA.99.062337},
 volume = {99},
 year = {2019}
}

@article{PhysRevLett.128.110503,
 author = {Hastrup, Jacob and Park, Kimin and Brask, Jonatan Bohr and Filip, Radim and Andersen, Ulrik Lund},
 doi = {10.1103/PhysRevLett.128.110503},
 issue = {11},
 journal = {Phys. Rev. Lett.},
 month = {Mar},
 numpages = {5},
 pages = {110503},
 publisher = {American Physical Society},
 title = {Universal Unitary Transfer of Continuous-Variable Quantum States into a Few Qubits},
 url = {https://link.aps.org/doi/10.1103/PhysRevLett.128.110503},
 volume = {128},
 year = {2022}
}

@article{PhysRevLett.130.090602,
 author = {Chabaud, Ulysse and Walschaers, Mattia},
 doi = {10.1103/PhysRevLett.130.090602},
 issue = {9},
 journal = {Phys. Rev. Lett.},
 month = {Mar},
 numpages = {7},
 pages = {090602},
 publisher = {American Physical Society},
 title = {Resources for Bosonic Quantum Computational Advantage},
 url = {https://link.aps.org/doi/10.1103/PhysRevLett.130.090602},
 volume = {130},
 year = {2023}
}

@article{PhysRevResearch.3.033018,
 author = {Chabaud, Ulysse and Ferrini, Giulia and Grosshans, Fr\'ed\'eric and Markham, Damian},
 doi = {10.1103/PhysRevResearch.3.033018},
 issue = {3},
 journal = {Phys. Rev. Res.},
 month = {Jul},
 numpages = {12},
 pages = {033018},
 publisher = {American Physical Society},
 title = {Classical simulation of Gaussian quantum circuits with non-Gaussian input states},
 url = {https://link.aps.org/doi/10.1103/PhysRevResearch.3.033018},
 volume = {3},
 year = {2021}
}

@article{preskill2018,
 author = {Preskill, John},
 doi = {10.22331/q-2018-08-06-79},
 journal = {Quantum},
 pages = {79},
 title = {Quantum {{Computing}} in the {{NISQ}} Era and Beyond},
 volume = {2},
 year = {2018}
}

@article{rahimi-keshari2016,
 author = {Rahimi-Keshari, Saleh and Ralph, Timothy C. and Caves, Carlton M.},
 doi = {10.1103/PhysRevX.6.021039},
 issue = {2},
 journal = {Phys. Rev. X},
 month = {Jun},
 numpages = {13},
 pages = {021039},
 publisher = {American Physical Society},
 title = {Sufficient Conditions for Efficient Classical Simulation of Quantum Optics},
 volume = {6},
 year = {2016}
}

@article{RevModPhys.86.153,
 author = {Georgescu, I. M. and Ashhab, S. and Nori, Franco},
 doi = {10.1103/RevModPhys.86.153},
 issue = {1},
 journal = {Rev. Mod. Phys.},
 month = {Mar},
 numpages = {33},
 pages = {153--185},
 publisher = {American Physical Society},
 title = {Quantum simulation},
 url = {https://link.aps.org/doi/10.1103/RevModPhys.86.153},
 volume = {86},
 year = {2014}
}

@article{seddon2021,
 author = {Seddon, James R. and Regula, Bartosz and Pashayan, Hakop and Ouyang, Yingkai and Campbell, Earl T.},
 doi = {10.1103/PRXQuantum.2.010345},
 issn = {2691-3399},
 journal = {PRX Quantum},
 langid = {english},
 month = {March},
 number = {1},
 pages = {010345},
 shorttitle = {Quantifying {{Quantum Speedups}}},
 title = {Quantifying {{Quantum Speedups}}: {{Improved Classical Simulation From Tighter Magic Monotones}}},
 volume = {2},
 year = {2021}
}

@book{serafini2017,
 address = {{Boca Raton, FL}},
 author = {Serafini, Alessio},
 doi = {https://doi.org/10.1201/9781315118727},
 isbn = {978-1-4822-4634-6},
 publisher = {{CRC Press, Taylor \& Francis Group}},
 title = {Quantum Continuous Variables : A Primer of Theoretical Methods},
 year = {2017}
}

@article{Shaw2024,
 author = {Shaw, Mackenzie H. and Doherty, Andrew C. and Grimsmo, Arne L.},
 doi = {10.1103/PRXQuantum.5.010331},
 issue = {1},
 journal = {PRX Quantum},
 month = {Feb},
 numpages = {36},
 pages = {010331},
 publisher = {American Physical Society},
 title = {Stabilizer Subsystem Decompositions for Single- and Multimode Gottesman-Kitaev-Preskill Codes},
 url = {https://link.aps.org/doi/10.1103/PRXQuantum.5.010331},
 volume = {5},
 year = {2024}
}

@article{tacchino2020,
 author = {Tacchino, Francesco and Chiesa, Alessandro and Carretta, Stefano and Gerace, Dario},
 doi = {https://doi.org/10.1002/qute.201900052},
 journal = {Advanced Quantum Technologies},
 number = {3},
 pages = {1900052},
 title = {Quantum Computers as Universal Quantum Simulators: State-of-the-Art and Perspectives},
 url = {https://advanced.onlinelibrary.wiley.com/doi/abs/10.1002/qute.201900052},
 volume = {3},
 year = {2020}
}

@article{upreti2025,
 author = {Upreti, Varun and Rudolph, Dorian and Chabaud, Ulysse},
 doi = {10.1038/s41534-026-01255-6},
 journal = {npj Quantum Information},
 publisher = {Nature Publishing Group UK London},
 title = {Bounding the computational power of bosonic systems},
 year = {2026}
}

@article{van-den-nest2010,
 author = {Van Den Nest, Maarten},
 doi = {10.26421/QIC10.3-4-6},
 journal = {Quantum Inf. Comput.},
 number = {3},
 pages = {258--271},
 publisher = {{Rinton Press, Incorporated Paramus, NJ}},
 title = {Classical Simulation of Quantum Computation, the {{Gottesman-Knill}} Theorem, and Slightly Beyond},
 volume = {10},
 year = {2010}
}

@article{veitch2012,
 author = {Victor Veitch and Christopher Ferrie and David Gross and Joseph Emerson},
 doi = {10.1088/1367-2630/14/11/113011},
 journal = {New J. Phys.},
 month = {nov},
 number = {11},
 pages = {113011},
 publisher = {{IOP} Publishing},
 title = {Negative quasi-probability as a resource for quantum computation},
 url = {https://doi.org/10.1088%2F1367-2630%2F14%2F11%2F113011},
 volume = {14},
 year = {2012}
}

@article{veitch2013,
 author = {Victor Veitch and Nathan Wiebe and Christopher Ferrie and Joseph Emerson},
 doi = {10.1088/1367-2630/15/1/013037},
 journal = {New J. Phys.},
 month = {jan},
 number = {1},
 pages = {013037},
 publisher = {{IOP} Publishing},
 title = {Efficient simulation scheme for a class of quantum optics experiments with non-negative Wigner representation},
 url = {https://doi.org/10.1088%2F1367-2630%2F15%2F1%2F013037},
 volume = {15},
 year = {2013}
}

@article{vidal2003,
 author = {Vidal, Guifr\'e},
 doi = {10.1103/PhysRevLett.91.147902},
 issue = {14},
 journal = {Phys. Rev. Lett.},
 month = {Oct},
 numpages = {4},
 pages = {147902},
 publisher = {American Physical Society},
 title = {Efficient Classical Simulation of Slightly Entangled Quantum Computations},
 url = {https://link.aps.org/doi/10.1103/PhysRevLett.91.147902},
 volume = {91},
 year = {2003}
}

@inbook{watrous2009,
 address = {New York, NY},
 author = {Watrous, John},
 booktitle = {Encyclopedia of Complexity and Systems Science},
 doi = {10.1007/978-0-387-30440-3_428},
 editor = {Meyers, Robert A.},
 isbn = {978-0-387-30440-3},
 pages = {7174--7201},
 publisher = {Springer New York},
 title = {Quantum Computational Complexity},
 url = {https://doi.org/10.1007/978-0-387-30440-3_428},
 year = {2009}
}

@book{wilde2017,
 author = {Wilde, Mark M.},
 doi = {10.1017/CBO9781139525343},
 edition = {2},
 place = {Cambridge},
 publisher = {Cambridge University Press},
 title = {Quantum Information Theory},
 year = {2017}
}

@misc{winter2017,
 archiveprefix = {arXiv},
 author = {Andreas Winter},
 eprint = {1712.10267},
 primaryclass = {quant-ph},
 title = {Energy-constrained diamond norm with applications to the uniform continuity of continuous variable channel capacities},
 url = {https://arxiv.org/abs/1712.10267},
 year = {2017}
}

@article{yang2023,
 author = {Yang, Zebo and Zolanvari, Maede and Jain, Raj},
 doi = {10.1109/COMST.2023.3254481},
 journal = {IEEE Communications Surveys \& Tutorials},
 number = {2},
 pages = {1059-1094},
 title = {A Survey of Important Issues in Quantum Computing and Communications},
 volume = {25},
 year = {2023}
}

@article{zohar2021,
 author = {Zohar, Erez},
 doi = {10.1098/rsta.2021.0069},
 issn = {1364-503X},
 journal = {Philos. Trans. R. Soc. A},
 number = {2216},
 pages = {20210069},
 title = {Quantum Simulation of Lattice Gauge Theories in More than One Space Dimension—Requirements, Challenges and Methods},
 url = {https://doi.org/10.1098/rsta.2021.0069},
 volume = {380},
 year = {2021}
}
\appendix

\begin{widetext}

\section{Energy analysis of CV operations}
\label{sec:appendix-energy}
In this Appendix, we will derive the effect of different CV operations on the energy of the state. First, note that passive operations do not increase the energy of a state~\cite{braunstein2005}. Passive operations include rotations and beam splitters. Active operations are those such as squeezing, displacements, shear, and the cubic phase gate. We will now derive the effect of each of these operations on the energy of a general state.

\subsection{Proof of maximum energy increase of squeezing}
\begin{lemma}
    \label{lemma:energy-squeezing}
    The energy of a state $\hat\rho$ after applying the squeezing operator with parameter $r$ is bound as $e^{-2|r|} E_{\hat \rho}\leq E_{\hat{S}(r)\hat{\rho}\hat{S}^\dagger(r)} \leq e^{2|r|} E_{\hat \rho}$.
\end{lemma}
\begin{proof}
The energy of the squeezed state can be written as
\begin{align}
    E_{\hat{S}(r)\hat{\rho}\hat{S}^\dagger(r)} & = \frac{1}{2}\left(\langle\hat{q}^2\rangle_{\hat{S}(r)\hat{\rho}\hat{S}^\dagger(r)}+\langle\hat{p}^2\rangle_{\hat{S}\hat{\rho}\hat{S}^\dagger(r)}\right) \notag\\
    & = \frac{1}{2}\left(e^{2r}\langle\hat{q}^2\rangle_{\hat{\rho}}+e^{-2r}\langle\hat{p}^2\rangle_{\hat{\rho}}\right) \label{eq:par squeezed energy}
\end{align}
If we first consider the case when $r\geq0$, then we can see that 
\begin{align}
    E_{\hat{S}(r)\hat{\rho}\hat{S}^\dagger(r)} & \leq \frac{1}{2}\left(e^{2r}\langle\hat{q}^2\rangle_{\hat{\rho}}+e^{2r}\langle\hat{p}^2\rangle_{\hat{\rho}}\right) \notag\\
    & = e^{2r}E_{\hat{\rho}},
\end{align}
where we have used $\langle\hat{q}^2\rangle_{\hat{\rho}},\langle\hat{p}^2\rangle_{\hat{\rho}}\geq0$. We can also obtain a lower bound of the energy from Eq.~\eqref{eq:par squeezed energy} as
\begin{align*}
    E_{\hat{S}(r)\hat{\rho}\hat{S}^\dagger(r)} & \geq \frac{1}{2}\left(e^{-2r}\langle\hat{q}^2\rangle_{\hat{\rho}}+e^{-2r}\langle\hat{p}^2\rangle_{\hat{\rho}}\right) \notag \\
    & = e^{-2r}E_{\hat{\rho}}
\end{align*}
Note that the case of $r<0$ is equivalent to performing a rotation before and after applying positive squeezing. Since the energy does not change when rotating the state, we immediately see that
\begin{align}
    e^{-2|r|} E_{\hat \rho}\leq E_{\hat{S}(r)\hat{\rho}\hat{S}^\dagger(r)} \leq e^{2|r|} E_{\hat \rho}.
\end{align}
\end{proof}
\subsection{Proof of maximum energy increase of shear gate}
\begin{lemma}
    \label{lemma:energy-shear}
    The energy of a state $\hat\rho$ after applying the shear operator with parameter $s$ is bound by $E_{\hat P(s)\hat \rho \hat P^\dagger(s)}\leq (1+s)^2 E_{\hat \rho}$.
\end{lemma}

\begin{proof}
    For this proof, we will use the decomposition of the shear gate from Ref.~\cite{Kalajdzievski_2021}
\begin{equation}
    \label{eq:squeezing-to-shear}
    \hat{P}(s)=e^{i\frac{s}{2}\hat{q}^2} = \hat{R}(\theta)\hat{S}(r)\hat{R}(\theta')
\end{equation}
where
\begin{equation}
\label{eq: decomp params}
    r=\text{arcsinh}\left(\frac{s}{2}\right)=\ln\left[\frac{s}{2}+\sqrt{\frac{s^2}{4}+1}\right],\quad\tan(\theta)=e^r,\quad\theta'=\theta-\frac{\pi}{2}.
\end{equation}
Therefore, we have
\begin{align}
    E_{\hat{P}(s)\hat{\rho}\hat{P}(s)^\dagger} & = E_{\hat{R}(\theta)\hat{S}(r)\hat{R}(\theta')\hat{\rho}\hat{R}^\dagger(\theta')\hat{S}^\dagger(r)\hat{R}^\dagger(\theta)} \notag\\
    & = E_{\hat{S}(r)\hat{R}(\theta')\hat{\rho}\hat{R}^\dagger(\theta')\hat{S}^\dagger(r)} \notag\\
    & \leq e^{2r}E_{\hat{R}(\theta')\hat{\rho}\hat{R}^\dagger(\theta')} \notag\\
    & =  e^{2r}E_{\hat{\rho}} \notag\\
    & = \frac{1}{4}\left(s+\sqrt{s^2+4}\right)^2E_{\hat{\rho}} \notag\\
    & \leq \frac{1}{4}(2+2s)^2E_{\hat{\rho}} \notag\\
    & = (1+s)^2E_{\hat{\rho}},
\end{align}
where Lemma~\ref{lemma:energy-squeezing} has been used in the third line, $\left(s+\sqrt{s^2+4}\right)^2\leq(2+2s)^2$ in the fifth line and that rotations does not change the energy. 
\end{proof}
\subsection{Proof of maximum energy increase of controlled-\texorpdfstring{$Z$}{Z} gate}
\begin{lemma}
\label{lemma:energy-controlled-Z}
    The energy of a state $\hat\rho$ after applying the controlled-$Z$ operator with parameter $s$ is bound by $E_{\CZGate*(s)\hat \rho \CZGate*^\dagger(s)}\leq (s+1)^4E_{\hat \rho}$.
\end{lemma}
\begin{proof}
\label{sec:appendix-proof-energy-control}
For this proof, we will again use an decomposition provided in Ref.~\cite{Kalajdzievski_2021} 
\begin{equation}
    \CZGate*(s)=e^{is\hat{q}_k\hat{q}_l}=\hat{F}_k\BSGate*(\theta)\hat{F}_k^\dagger \hat{S}_k(r)\hat{S}_l(r) \hat{F}_k\BSGate*(\theta')\hat{F}_k^\dagger,
\end{equation}
where the parameters are the same as in Eq.~\eqref{eq: decomp params}.
We see immediately that the state $\hat \rho'$ after the first beam splitter and Fourier transforms has the same energy as $\hat\rho$. The state $\hat\rho''$ after the first squeezing operation has maximum energy $e^{2r}E_{\hat \rho}$, the state $\hat\rho'''$ after the second squeezing operation has maximum energy $e^{4r}E_{\hat\rho}$, which comes from Lemma~\ref{lemma:energy-squeezing}. This state has the same energy as the state after applying the controlled-$Z$ operation. Therefore, we bound the energy as
\begin{align}
    e^{4r}E_{\hat\rho}=&\frac{1}{16}\left(s+\sqrt{4+s^2}\right)^4E_{\hat \rho}\nonumber\\
    \leq & \frac{1}{16}(2s+2)^4E_{\hat \rho}\nonumber\\
    = & (s+1)^4E_{\hat \rho}.
\end{align}
\end{proof}
\section{Maximum possible parameters of operations}
\label{sec:Max-energy}
Here, we identify the maximum possible parameters of squeezing and displacement by analyzing the effect of each operation on the energy of an arbitrary state.
\subsection{Maximum value of squeezing}\label{proof:min squeezing}
With the maximum possible energy of a state being $E^*$, we can consider what happens to the energy of the lowest energy state (i.e., the vacuum state) when squeezed by $r$. We can bound the squeezing parameter as the maximum value of squeezing that can be applied before reaching a level of energy that is too high for the system to contain using Lemma~\ref{lemma:energy-squeezing}.
Solving for the case that the initial state $\hat \rho$ has the lowest possible energy, $\frac 1 2$, and $\hat\rho'$ has the maximum energy, $E_{\hat \rho'}=E^*$, we find
\begin{align}\label{eq:max squeezing}
    &\frac{1}{2}e^{2|r|}\leq E^*
    \nonumber\\
    \implies & e^{|r|}\leq \sqrt{2E^*}.
\end{align}

\subsection{Proof of maximum cubicity of cubic phase gates}
\begin{lemma}
    \label{lemma:energy-cubic-phase}
    The maximum value of the cubic phase gate parameter $\gamma$ is bound in terms of the maximum energy of the system as $\gamma\leq 8 E^{*3/2}$.
\end{lemma}

\begin{proof}
\begin{align*}
        E_{\hat{C}(\gamma)\hat{\rho}\hat{C}^\dagger(\gamma)} & = \frac{1}{2}\left(\langle\hat{q}^2\rangle_{\hat{C}(\gamma)\hat{\rho}\hat{C}^\dagger(\gamma)}+\langle\hat{p}^2\rangle_{\hat{C}(\gamma)\hat{\rho}\hat{C}^\dagger(\gamma)}\right) \\
        & = \frac{1}{2}\left(\langle\hat{q}^2\rangle_{\hat{\rho}}+\langle\left(\hat{p}+3\gamma\hat{q}\right)^2\rangle_{\hat{\rho}}\right) \\
        & = \frac{1}{2}\left(\langle\hat{q}^2\rangle_{\hat{\rho}}+\langle\hat{p}^2\rangle_{\hat{\rho}}+3\gamma\langle\hat{q}^2\hat{p}+\hat{p}\hat{q}^2\rangle_{\hat{\rho}}+9\gamma^2\langle\hat{q}^4\rangle_{\hat{\rho}}\right) \\
        & = E_{\hat{\rho}} + \frac 3 2 \gamma \left\langle \hat q^2\hat p+\hat p\hat q^2\right\rangle_{\hat \rho}+\frac{9}{2}\gamma^2\left\langle\hat{q}^4\right\rangle_{\hat \rho}
\end{align*}
Hence, we have that
\begin{align}
\label{eq:bound-energy-cps}
     \frac 3 2 \gamma \left\langle \hat q^2\hat p+\hat p\hat q^2\right\rangle_{\hat \rho}+\frac{9}{2}\gamma^2\left\langle\hat{q}^4\right\rangle_{\hat \rho}\leq E^*-E_{\hat\rho} \leq E^* 
\end{align}
Since $\left\langle \hat q^4\right\rangle_{\hat \rho}\geq 0$, the inequality represents a parabola that opens upwards. 
Taking equality, we find solutions to the equation as
\begin{align}
    \gamma^*=\frac{-\frac 3 2  \left\langle \hat q^2\hat p+\hat p\hat q^2\right\rangle_{\hat \rho}\pm \sqrt{\frac{9}{4}\left\langle \hat q^2\hat p+\hat p\hat q^2\right\rangle_{\hat \rho}^2+4\frac{9}{2}E^* \left\langle\hat{q}^4\right\rangle_{\hat \rho} }}{9\left\langle\hat{q}^4\right\rangle_{\hat \rho}}.
\end{align}
Hence,
\begin{align}
    \gamma\leq& \frac{-\frac 3 2  \left\langle \hat q^2\hat p+\hat p\hat q^2\right\rangle_{\hat \rho}+ \sqrt{\frac{9}{4}\left\langle \hat q^2\hat p+\hat p\hat q^2\right\rangle_{\hat \rho}^2+18E^* \left\langle\hat{q}^4\right\rangle_{\hat \rho} }}{9\left\langle\hat{q}^4\right\rangle_{\hat \rho}}\nonumber\\
    \leq& \frac{3  \sqrt{\langle q^4\rangle_{\hat \rho}\langle p^2\rangle_{\hat \rho}}+ \sqrt{9\langle q^4\rangle_{\hat \rho}\langle p^2\rangle_{\hat \rho}+18E^* \left\langle\hat{q}^4\right\rangle_{\hat \rho} }}{9\left\langle\hat{q}^4\right\rangle_{\hat \rho}}\nonumber\\
    =& \frac{3  \sqrt{\langle p^2\rangle_{\hat \rho}}+ \sqrt{9\langle p^2\rangle_{\hat \rho}+18E^* }}{9\sqrt{\left\langle\hat{q}^4\right\rangle_{\hat \rho}}}.
\end{align}
We find a lower bound on the expectation value of $\hat q^4$ as
\begin{align}
    \langle \hat q^4\rangle_{\hat \rho} \geq \langle \hat q^2\rangle_{\hat \rho}^2,
\end{align}
where we have used Jensen's inequality. Using the Heisenberg uncertainty relation
\begin{align}
    &(\langle \hat q^2\rangle -\langle \hat q\rangle^2)(\langle \hat p^2\rangle -\langle \hat p\rangle^2)\geq \frac 1 4\nonumber\\
\end{align}
we find that 
\begin{equation*}
    \langle \hat q^2\rangle \langle \hat p^2\rangle\geq \frac 1 4\nonumber,
\end{equation*}
which gives the inequality
\begin{equation*}
    \langle \hat q^2\rangle\geq \frac{1}{4 \langle \hat p^2\rangle}.
\end{equation*}
Therefore,
\begin{align}
    \gamma\leq&  \frac{3  \sqrt{\langle p^2\rangle_{\hat \rho}}+ \sqrt{9\langle p^2\rangle_{\hat \rho}+18E^* }}{9\left\langle\hat{q}^2\right\rangle_{\hat \rho}}\nonumber\\
    \leq&  \frac{4}{9}\left(3  \sqrt{\langle p^2\rangle_{\hat \rho}}+ \sqrt{9\langle p^2\rangle_{\hat \rho}+18E^* }\right)\langle \hat p^2\rangle\nonumber\\
    \leq&  \frac{8}{9}\left(3  \sqrt{2E^*}+ 6\sqrt{E^* }\right)E^*\nonumber\\
    \leq&  8E^{*3/2}
\end{align}
where we have used that $\langle \hat p^2\rangle\leq 2E^*$.
\end{proof}

\section{Approximation error of circuit components}
\label{appendix: errors}
In this Appendix, we bound the approximation error of various operations when mapping from CV to DV. We begin with a formal explanation of the error that we bound, and then we provide proofs of the errors for each of the operations. Appendix~\ref{sec:appendix-proof-error-fourier-transform} provides the proof to Lemma~\ref{lemma:error-Fourier-transform}. Appendix~\ref{sec:error-gen-phase} introduces a general phase gate which will then be used to prove the error of operators that are polynomial in $\hat{\mathbf q}$. Appendix~\ref{sec:appendix-displacement-error} provides the proof to Lemma~\ref{lemma:error-displacement}. Appendix~\ref{sec:appendix-shear-gate} provides the proof to Lemma~\ref{lemma:error-shear-gate}. Appendix~\ref{sec:appendix-cubic-phase-gate} provides the proof to Lemma~\ref{lemma:error-cubic-phase}. Appendix~\ref{sec:appendix-error-controlled-z} provides the proof to Lemma~\ref{lemma:error-controlled-z}. Next, we use these gates to construct other CV gates and derive the error of applying each of these new gates in terms of the previous. Appendix~\ref{sec:appendix-rotation} provides the proof to Lemma~\ref{lemma:error-rotation}, Appendix~\ref{sec:appendix-beamsplitter} provides the proof to Lemma~\ref{lemma:error-beamsplitter}, and Appendix~\ref{sec:appendix-squeezing} provides the proof to Lemma~\ref{lemma:error-squeezing}.

In general the SSD of a CV gate $\hat U$ acting on a state $\hat{\rho} \in \mathcal U_E$ can be expressed as $\mathrm{Tr}_S (\hat U \hat\rho \hat U^\dagger)$, while the SSD on the state, followed by performing the corresponding DV gate $\hat U_d$ can be expressed as $\hat U_d \mathrm{Tr}_S (\hat\rho) \hat U_d^\dagger$. Hence the trace distance that arrises when approximating $\mathrm{Tr}_S (\hat U \hat\rho \hat U^\dagger)$ with $\hat U_d \mathrm{Tr}_S (\hat\rho) \hat U_d^\dagger$ can be expressed as:
\begin{align}
    \frac{1}{2}\left\| \hat U_d \mathrm{Tr}_S (\hat\rho) \hat U_d^\dagger - \mathrm{Tr}_S (\hat U \hat\rho \hat U^\dagger) \right\|_1.
\end{align}
Since the trace distance is unitarily invariant, we can rewrite the expression as
\begin{align}
    \frac{1}{2}\left\| \hat U_d \text{Tr}_S (\hat\rho) \hat U_d^\dagger - \text{Tr}_S (\hat U \hat\rho \hat U^\dagger) \right\|_1 &=  \frac{1}{2}\left\| \hat U_d^\dagger \left(\hat U_d \text{Tr}_S (\hat\rho) \hat U_d^\dagger - \text{Tr}_S (\hat U \hat\rho \hat U^\dagger) \right) \hat U_d \right\|_1\nonumber\\
    \label{eq:ssderror}
    &= \frac{1}{2}\left\| \text{Tr}_S (\hat\rho) - \hat U_d^\dagger \text{Tr}_S (\hat U \hat\rho \hat U^\dagger) \hat U_d \right\|_1.
\end{align}
\subsection{Proof of error of Fourier transform}
\label{sec:appendix-proof-error-fourier-transform}
In this Section, we will prove Lemma~\ref{lemma:error-Fourier-transform}.
\begin{proof}
    For the CV Fourier transform $\FourierTransform*$, the corresponding DV Fourier transform $\FourierTransform$ is equivalent, such that the Fourier transform is preserved under the action of the SSD. 
Expanding the second term in Eq.~(\ref{eq:ssderror}), gives
\begin{align}
		  \hat F_d^\dagger \text{Tr}_S (\hat F \hat\rho \hat F^\dagger) \hat F_d  &=  \frac{1}{\ell}\hat F_d^\dagger  \int_{\mathbb T_{\ell}^{2}}  \text{d} \mathbf{t} \hat{\Pi} \hat{D}(-\mathbf{t}) \hat F \hat\rho \hat F^\dagger \hat{D}^{\dagger}(-\mathbf{t}) \hat{\Pi}^\dagger \hat F_d.
\end{align}
Now, since applying the CV Fourier transform to the GKP projector results in
\begin{align}
    \hat F \Pi^{\dagger} &=  e^{i\frac{\pi}{4}(\hat q^2+\hat p^2)} \sum_{a=0}^{d-1} \ketbra{a_{GKP}^{(d)}}{a}\notag\\
    &= \frac{1}{\sqrt{d}} \sum_{a,b=0}^{d-1} \omega^{ab}_d \ketbra{b_{\text{GKP}}^{(d)}}{a},
\end{align}
it follows that $\hat F_d \Pi = \Pi \hat F$. We also have $\hat F^\dagger \hat{D}(\mathbf{t}) \hat F = e^{i(t_q \hat q - t_p \hat p)}$. Using this leaves us with
\begin{align}
		  \frac{1}{\ell}\hat F_d^\dagger  \int_{\mathbb T_{\ell}^{2}} \text{d} \mathbf{t} \hat{\Pi} \hat{D}(\mathbf{t}) \hat F \hat\rho \hat F^\dagger \hat{D}^{\dagger}(\mathbf{t}) \hat{\Pi}^\dagger \hat F_d
          &=  \frac{1}{\ell}  \int_{\mathbb T_{\ell}^{2}} \text{d} \mathbf{t} \hat{\Pi} \hat F^\dagger \hat{D}(\mathbf{t}) \hat F \hat\rho \hat F^\dagger \hat{D}^{\dagger}(\mathbf{t}) \hat F \hat{\Pi}^\dagger \notag\\
          &=\frac{1}{\ell}  \int_{\mathbb T_{\ell}^{2}} \text{d} \mathbf{t} \hat{\Pi} e^{i(t_q \hat q - t_p \hat p)} \hat\rho e^{-i(t_q \hat q - t_p \hat p)} \hat{\Pi}^\dagger.
\end{align}
Since the integration bounds are identical and symmetric for $t_q$ and $t_p$, we can relabel them such that $t_q \to -t_p$ and $t_p \to -t_q$.
\begin{align}
          \frac{1}{\ell}  \int_{\mathbb T_{\ell}^{2}} \text{d} \mathbf{t} \, \hat{\Pi} \hat e^{i(-t_p \hat q + t_q \hat p)} \hat\rho e^{-i(-t_p \hat q + t_q \hat p)} \hat{\Pi}^\dagger &=\frac{1}{\ell}  \int_{\mathbb T_{\ell}^{2}} \text{d} \mathbf{t} \, \hat{\Pi} \hat{D}(\mathbf{t}) \hat\rho \hat{D}^{\dagger}(\mathbf{t}) \hat{\Pi}^\dagger \notag \\
          &= \text{Tr}_S (\hat\rho).
\end{align}
This means that
\begin{align}
    \frac{1}{2}\left\| \text{Tr}_S (\hat\rho) - \hat F_d^\dagger \text{Tr}_S (\hat F \hat\rho \hat F^\dagger) \hat F_d \right\|_1 &= \frac{1}{2}\left\| \text{Tr}_S (\hat\rho) - \text{Tr}_S (\hat\rho)  \right\|_1 = 0 .
\end{align}
\end{proof}

\subsection{Error of generalized phase gate}
\label{sec:error-gen-phase}
To obtain the error for the simulation of displacement, the shear gate, the controlled-$Z$ gate, and the cubic phase gate, we introduce what we call the generalized phase gate. In CV, this gate introduces a phase depending on a polynomial of the position operators. In DV, the phase depends on a polynomial of the computational basis states. Formally, we define the generalized CV phase gate acting on an ensemble of $n$ modes
		\begin{align}
        \label{eq:generalphase}
			\hat P^f=e^{if\left(\hat{\mathbf{q}}\right)}
		\end{align}
		where $f:\mathbb R^n\rightarrow\mathbb R$ is a polynomial function. 
		Whenever $f$ is periodic, fulfilling the condition
		\begin{align}\label{eq:QuasiPeriodicity}
			f\left(\mathbf h + \sqrt{2\pi d}\mathbf k \right)\equiv f\left(\mathbf h\right) \mod 2\pi\qquad \forall \mathbf k\in\mathbb Z^n, \mathbf h\in\mathbb R^n,
		\end{align}
        there is a well-defined corresponding approximate discrete-variable operation. 
		This can be seen in terms of the action of the gate on GKP encoded qudits,
        \begin{align}
   \hat P^f\ket{\mathbf{a}_{\text{GKP}}} &= e^{if\left(\hat{\mathbf{q}}\right)}\ket{\mathbf{a}_{\text{GKP}}}\notag \\
   &= e^{if\left(\hat{\mathbf{q}}\right)} \sum_{\mathbf{m} \in \mathbb{Z}^n} \ket{\hat{q}=\sqrt{\frac{2\pi}{d}} (\mathbf{m}d+\mathbf{a})}\notag \\
   &= \sum_{\mathbf{m} \in \mathbb{Z}^n} e^{i f\left(\sqrt{\frac{2\pi}{d}} (\mathbf{m}d+\mathbf{a})\right)} \ket{\hat{q}=\sqrt{\frac{2\pi}{d}} (\mathbf{m}d+\mathbf{a})}\notag \\
   &= \sum_{\mathbf{m} \in \mathbb{Z}^n} e^{i f\left(\sqrt{\frac{2\pi}{d}}\mathbf{a} + \sqrt{2\pi d} \mathbf{m}\right)} \ket{\hat{q}=\sqrt{\frac{2\pi}{d}} (\mathbf{m}d+\mathbf{a})}\notag \\
   &= \sum_{\mathbf{m} \in \mathbb{Z}^n} e^{i f\left(\sqrt{\frac{2\pi}{d}}\mathbf{a}\right)} \ket{\hat{q}=\sqrt{\frac{2\pi}{d}} (\mathbf{m}d+\mathbf{a})}\notag \\
   &= e^{i f\left(\sqrt{\frac{2\pi}{d}}\mathbf{a}\right)} \ket{\mathbf{a}_{\text{GKP}}},
    \end{align}
     It follows that the GKP projector can be used to map $\hat P^f$ to a corresponding DV operation
		\begin{align}
			\hat \Pi\hat P^f=\hat P_d^f\hat \Pi
		\end{align}
		where $\hat P_d^f$, acting on $n$ $d$-dimensional qudits, is defined as
		\begin{align}
			\hat P_d^f=\sum_{\mathbf a\in \mathbb Z_d^n}e^{if\left(\sqrt{\frac{2\pi}{d}}\mathbf{a}\right)} \ketbra{\mathbf a}{\mathbf a}.
		\end{align}
		Moreover, the restriction of Eq.~\eqref{eq:QuasiPeriodicity} can be lifted by modifying the construction in two ways. First, by defining the DV operator as
		\begin{align}
			\hat P_d^f=\sum_{\mathbf a\in \mathbb Z_d^n} e^{if\left(\sqrt{\frac{2\pi}{d}}\left\{\mathbf{a}\right\}_d\right)} \ketbra{\mathbf a}{\mathbf a},
		\end{align}
        whereby $\left\{a\right\}_d=\left(a+\lfloor \frac d 2 \rfloor\right)-\lfloor \frac d 2 \rfloor$, and $\left\{\mathbf a\right\}_d=\left(\left\{a_1\right\}_d,\dots,\left\{a_n\right\}_d\right)^T$. Second, by including the projector $\hat \Lambda_{d}$ resulting in
		\begin{align}
			\hat \Pi\hat \Lambda_{d}\hat P^f=&\hat \Pi \hat P^f \Lambda_d\nonumber\\
            =& \sum_{\mathbf a\in \mathbb Z_d^n}\ketbra{\mathbf a}{\mathbf a_{\text{GKP}}}\hat\Lambda_d e^{if(\hat{\mathbf q})}\nonumber\\
            =& \sum_{\mathbf a\in \mathbb Z_d^n}\ketbra{\mathbf a}{\hat q=\left\{\mathbf a\right\}_d\ell}e^{if(\hat{\mathbf q})}\nonumber\\
            =&\hat P_d^f\hat \Pi\hat \Lambda_{d}.
		\end{align}
In the following, we will use these three commutation relations:
		\begin{align}
			e^{if\left(\mathbf{\hat q}\right)}\hat D(\mathbf{t}) &= \hat D(\mathbf{t})e^{if\left(\mathbf{\hat q}-\mathbf t_q\right)},\\
			\hat \Lambda_{d}e^{if\left(\mathbf{\hat q}\right)}&=e^{if\left(\mathbf{\hat q}\right)}\hat \Lambda_{d},
		\end{align}
        and
        \begin{align}
			\hat \Pi\hat \Lambda_{d}\hat D(\mathbf{t}) &= \hat \Pi\hat D(\mathbf{t})\hat \Lambda_{d},
        \end{align}
        assuming that the value of each component of $\mathbf t$ is in the range $(\ell/2,\ell/2]$.
        
Expanding the second term in Eq.~(\ref{eq:ssderror}),  while also using $\hat \rho=\hat \Lambda_{d}\hat \rho\hat \Lambda_{d}$, gives us
\begin{align}
			\hat P_d^{f \dagger} \mathrm{Tr}_S (\hat P^f \hat\rho \hat P^{f \dagger}) \hat P_d^f  &=  \frac{1}{\ell^n}\hat P_d^{f \dagger}  \int_{\mathbb T_{\ell}^{2n}} \text{d} \mathbf{t} \hat{\Pi} \hat{D}(\mathbf{t}) \hat P^f \hat\rho \hat P^{f \dagger} \hat{D}^{\dagger}(\mathbf{t}) \hat{\Pi}^\dagger \hat P_d^f \nonumber\\
			&=  \frac{1}{\ell^n} \int_{\mathbb T_{\ell}^{2n}} \text{d} \mathbf{t} \hat P_d^{f \dagger} \hat{\Pi}\hat \Lambda_{d} \hat{D}(\mathbf{t}) \hat P^f \hat\rho \hat P^{f \dagger} \hat{D}^{\dagger}(\mathbf{t}) \hat \Lambda_{d}\hat{\Pi}^\dagger \hat P_d^f\nonumber \\
			&=  \frac{1}{\ell^n} \int_{\mathbb T_{\ell}^{2n}} \text{d} \mathbf{t} \hat{\Pi}\hat \Lambda_{d}\hat P^{f \dagger}\hat{D}(\mathbf{t}) \hat P^f \hat\rho \hat P^{f \dagger} \hat{D}^{\dagger}(\mathbf{t}) \hat P^f \hat \Lambda_{d}\hat{\Pi}^\dagger \nonumber\\
			&=  \frac{1}{\ell^n} \int_{\mathbb T_{\ell}^{2n}} \text{d} \mathbf{t} \hat{\Pi} \hat{D}(\mathbf{t}) e^{if\left(\hat{\mathbf{q}}\right)-if\left(\hat{\mathbf{q}}-\mathbf t_q\right)} \hat\rho e^{-if\left(\hat{\mathbf{q}}\right)+if\left(\hat{\mathbf{q}}-\mathbf t_q\right)}\hat{D}^{\dagger}(\mathbf{t}) \hat{\Pi}^\dagger.
		\end{align}
Now, from this, we see that
\begin{align}
			&\left\| \hat P_d^f \text{Tr}_S (\hat\rho) \hat P_d^{f \dagger} - \text{Tr}_S (\hat P^f \hat\rho \hat P^{f \dagger}) \right\|_1 \nonumber\\
			=&\frac{1}{\ell^n}  \left\|  \int_{\mathbb T_{\ell}^{2n}} \text{d} \mathbf{t} \hat{\Pi}\hat{D}(\mathbf{t}) \left(\hat \rho - e^{if\left(\hat{\mathbf{q}}\right)-if\left(\hat{\mathbf{q}}-\mathbf t_q\right)} \hat\rho e^{-if\left(\hat{\mathbf{q}}\right)+if\left(\hat{\mathbf{q}}-\mathbf t_q\right)}\right) \hat{D}^{\dagger}(\mathbf{t})\hat{\Pi}^\dagger \right\|_1 \nonumber\\
			\leq&\frac{1}{\ell^n} \int_{\mathbb T_{\ell}^{2n}} \text{d} \mathbf{t}  \left\|  \hat{\Pi}\hat{D}(\mathbf{t}) \left(\hat \rho - e^{if\left(\hat{\mathbf{q}}\right)-if\left(\hat{\mathbf{q}}-\mathbf t_q\right)} \hat\rho e^{-if\left(\hat{\mathbf{q}}\right)+if\left(\hat{\mathbf{q}}-\mathbf t_q\right)}\right) \hat{D}^{\dagger}(\mathbf{t})\hat{\Pi}^\dagger \right\|_1\nonumber \\
			\leq&\frac{1}{\ell^n } \int_{\mathbb T_{\ell}^{2n}} \text{d} \mathbf{t}  \left\|\hat\Pi\right\|_\infty\left\|  \hat \rho - e^{if\left(\hat{\mathbf{q}}\right)-if\left(\hat{\mathbf{q}}-\mathbf t_q\right)} \hat\rho e^{-if\left(\hat{\mathbf{q}}\right)+if\left(\hat{\mathbf{q}}-\mathbf t_q\right)}\right\|_1 \left\|\hat\Pi^\dagger\right\|_\infty\nonumber\\
			\leq&\frac{1}{\ell^n} \int_{\mathbb T_{\ell}^{2n}} \text{d} \mathbf{t} \left\|  \hat \rho - e^{if\left(\hat{\mathbf{q}}\right)-if\left(\hat{\mathbf{q}}-\mathbf t_q\right)} \hat\rho e^{-if\left(\hat{\mathbf{q}}\right)+if\left(\hat{\mathbf{q}}-\mathbf t_q\right)}\right\|_1.
\end{align}
In the last line, we have used the fact that the operator norm of both $\hat \Pi$ and $\hat \Pi^\dagger$ is less than or equal to one:
		\begin{align}\label{eq:OperatorBound}
			0\leq \left\|\hat\Pi\right\|_\infty,\left\|\hat\Pi^\dagger\right\|_\infty\leq1,
		\end{align}
where the operator norm for an operator $\hat A$ that maps $\mathcal A\to \mathcal B$ is defined by
		\begin{align}
			\| \hat A \|_\infty = \inf\{c\geq  0 : \|\hat A \ket \psi \| \leq c \|\ket{\psi}\|\quad  \forall\quad \ket\psi \in \mathcal B\}
		\end{align}
with the Euclidean norm $\|\cdot\|$. For a general operator $\hat A$ we have that
		\begin{align}
			\|\hat A \ket\psi \|^2 =\bra\psi \hat A^\dagger\hat A \ket\psi
			=\bra\psi (\hat A^\dagger\hat A \ket\psi)
			\leq\|\ket\psi\| \cdot \|\hat A^\dagger\hat A \ket \psi \|
			=\|\hat A^\dagger\hat A \ket \psi \|
		\end{align}
where we used the Cauchy-Schwarz inequality. This implies that
		\begin{align}
			\|\hat \Pi\ket \psi \|^2&\leq \|\hat \Pi^\dagger\hat \Pi\ket \psi \|\leq \sqrt{\|\hat \Pi^\dagger\hat \Pi\ket \psi \|}
		\end{align}
        and
        \begin{align}
			\|\hat \Pi^\dagger\ket \psi \|^2&\leq \|\hat \Pi\hat \Pi^\dagger\ket \psi \|\leq \sqrt{\|\hat \Pi\hat \Pi^\dagger\ket \psi \|},
        \end{align}
and as a consequence Eq.~\eqref{eq:OperatorBound}.

Now we use the same technique as in Ref.~\cite{winter2017} to bound
		\begin{equation}
			\left\| \hat\rho  -  e^{i\hat{H}} \hat\rho  e^{ -i \hat{H}} \right\|_1 \\
			\leq \left\|\left[ \hat\rho  , \hat{H} \right] \right\|_1.
		\end{equation}
This bound arises because of the following. For
	   \begin{align}
			\left\| \hat{\rho} - \hat{\rho}' \right\|_1,
		\end{align}
where $\hat{\rho}'$ is a slightly evolved state, given by:
		\begin{align}
			\hat{\rho}' = e^{i\hat H} \hat{\rho} e^{-i\hat H}.
		\end{align}
		
We have $e^{i\hat H t} = \hat \rho (t)$ and can express $\hat \rho= \hat \rho (0)$  and $\hat \rho'= \hat \rho (1)$. Due to unitary invariance of the trace norm, we have:
		\begin{align}
			\left\| \hat{\rho} - \hat{\rho}(1) \right\|_1 & \leq \sum_{j=0}^{N-1} \left\| \hat{\rho}\left(\frac{j+1}{N}\right) - \hat{\rho}\left(\frac{j}{N}\right) \right\|_1 \nonumber\\
			&= N \left\| \hat{\rho}\left(\frac{1}{N}\right) - \hat{\rho}\left(0\right) \right\|_1,
		\end{align}
		and since
		\begin{align}
			\hat \rho \left( \frac{1}{N} \right) - \hat \rho = \frac{i}{N} [\hat \rho, \hat H] + \mathcal{O}\left(\frac{1}{N^2}\right),
		\end{align}
		when $N \to \infty$ we get
        \begin{align}
			\left\| \hat{\rho} - \hat{\rho}(1) \right\|_1 \leq \left\| [\hat \rho, \hat H] \right\|_1.
		\end{align}
        By setting $\hat{H} = f\left(\hat{\mathbf{q}}\right)-f\left(\hat{\mathbf{q}}-\mathbf t_q\right)$, we obtain
		\begin{align}
			\left\| \hat\rho - e^{if\left(\hat{\mathbf{q}}\right)-if\left(\hat{\mathbf{q}}-\mathbf t_q\right)} \hat\rho e^{-if\left(\hat{\mathbf{q}}\right)+if\left(\hat{\mathbf{q}}-\mathbf t_q\right)}\right\|_1
			\leq  \left\| [\hat \rho, f\left(\hat{\mathbf{q}}\right)-f\left(\hat{\mathbf{q}}-\mathbf t_q\right)] \right\|_1,
		\end{align}
        which gives us
        \begin{align}
			\left\| \hat P_d^f \text{Tr}_S (\hat\rho) \hat P_d^{f \dagger} - \text{Tr}_S (\hat P^f \hat\rho \hat P^{f \dagger}) \right\|_1
			\leq\frac{1}{\ell} \int_{\mathbb T_{\ell}^{2}} \text{d} \mathbf{t} \left\| [\hat \rho, f\left(\hat{\mathbf{q}}\right)-f\left(\hat{\mathbf{q}}-\mathbf t_q\right)] \right\|_1.
		\end{align}
        We can now use this expression to find bounds for different gates defined by specific polynomial functions $f(\mathbf{\hat q})$.

It will be useful to have a bound on the commutator between a state and a  single position operator. Using the fact that $\hat q_i = \frac{1}{\sqrt{2}} (\hat a_i + \hat a_i^\dagger)$, and using the triangle inequality we have:
		\begin{align}
			\left\| [\hat \rho, \hat q_i] \right\|_1 &= \left\| \frac{1}{\sqrt{2}} \left([\hat \rho, \hat a_i] + [\hat \rho, \hat a_i^\dagger]\right) \right\|_1   
			\leq  \frac{1}{\sqrt{2}} \left( \left\| \hat \rho \hat a_i \right\|_1 + \left\| \hat a_i \hat \rho \right\|_1 + \left\| \hat \rho \hat a_i^\dagger \right\|_1 + \left\| \hat a_i^\dagger \hat \rho \right\|_1 \right).
		\end{align}
		Now using Eq.~(\ref{eq:Cauchy-Schwarz}) we get
		\begin{align}
			\left\| \hat \rho \hat a_i  \right\|_1 =  \left\| \sqrt{\hat \rho} \sqrt{\hat \rho} \hat a_i  \right\|_1 & \leq \sqrt{\text{Tr}\left(\sqrt{\hat \rho} \sqrt{\hat \rho} \right)} \sqrt{\text{Tr}\left(\hat a_i^\dagger \sqrt{\hat \rho}  \sqrt{\hat \rho} \hat a_i\right)} \nonumber\\
			& = \sqrt{\text{Tr}\left(\hat \rho \right)} \sqrt{\text{Tr}\left(\hat a_i^\dagger \hat \rho \hat a_i\right)} \nonumber\\
			& = \sqrt{\text{Tr}\left(\hat a_i \hat a_i^\dagger \hat \rho \right)} \nonumber\\
			&= \sqrt{\langle\hat n_i \rangle_{\hat\rho} +1},
		\end{align}
		and similarly
		\begin{align}
			\left\| \hat a_i \hat \rho   \right\|_1 &\leq \sqrt{\langle\hat n_i \rangle_{\hat\rho}}, \notag\\
			\left\| \hat a_i^\dagger \hat \rho   \right\|_1 & \leq\sqrt{\langle\hat n_i \rangle_{\hat\rho} +1},\notag\\
			\left\|  \hat \rho \hat a_i^\dagger  \right\|_1 & \leq\sqrt{\langle\hat n_i \rangle_{\hat\rho}}.
		\end{align}
		The convenient bound $\sqrt{x}+\sqrt{x+1}\leq 2\sqrt{x+\tfrac 1 2}$ for all $x\geq 0$ gives us
		\begin{align}
			\left\| [\hat \rho, \hat q_i] \right\|_1 \leq 2\sqrt{2} \sqrt{\langle\hat n_i \rangle_{\hat\rho}+\tfrac{1}{2}}\leq 2\sqrt{2E_{\hat \rho}}.
		\end{align}

        \subsection{Proof of error of displacement}
        \label{sec:appendix-displacement-error}
		Using the results of the previous Section, we now look at error bounds for specific gates. We begin by providing a proof of Lemma~\ref{lemma:error-displacement}.
        \begin{proof}The single-mode $p$-displacement with $f(\hat q)=s\hat q$ which turns out to be exact as
		\begin{align}
			\left\| [\hat \rho, s\hat{q}-s(\hat{q}+ t_q)] \right\|_1=0.
		\end{align}
        Since the CV displacement, $\hat D(\mathbf{r})$, defined in Eq.~\ref{eq: displacement}, can be decomposed as $\hat Z(q) \hat F \hat Z(p) \hat F^\dagger$, up to a global phase, it follows that it is also exact.\end{proof}
        \subsection{Proof of error of shear gate}
        \label{sec:appendix-shear-gate}
		Next, we provide a proof of Lemma~\ref{lemma:error-shear-gate}.
        \begin{proof}
        The trace distance between a state with the single-mode shear gate applied before and after the SSD--- where the single-mode shear gate is defined as $f(\hat q)=\frac{s}{2}\hat q^2$---is bounded by
		\begin{align}
			\frac 1 2 \left\| \hat P_d^f \text{Tr}_S (\hat\rho) \hat P_d^{f \dagger} - \text{Tr}_S (\hat P^f \hat\rho \hat P^{f \dagger}) \right\|_1 
            \leq& \frac{1}{2\ell} \int_{\mathbb T_{\ell}^{2}} \text{d} \mathbf{t} \left\|  \left[\hat \rho, \frac s 2 \hat q^2 - \frac s 2 (\hat q - t_q)^2 \right] \right\|_1\nonumber\\
			=&\frac{1}{2\ell} \int_{\mathbb T_{\ell}^{2}} \text{d} \mathbf{t} \left\|  [\hat \rho,s t_q\hat q] \right\|_1\nonumber\\
			\leq&\frac{|s| \sqrt{2E_{\hat \rho}}}{\ell } \int_{\mathbb T_{\ell}^{2}} \text{d} \mathbf{t} |t_q|\nonumber\\
            =&\frac{|s| \sqrt{2E_{\hat \rho}}}{\ell}\ell \left(\frac{\ell}{2}\right)^2\nonumber\\
			=&|s| \frac{\pi}{d} \sqrt{\frac{E_\rho}{2}}.
		\end{align}
        \end{proof}
        \subsection{Proof of error of cubic phase gate}
        \label{sec:appendix-cubic-phase-gate}
        Here we provide a proof of Lemma~\ref{lemma:error-cubic-phase}.
        \begin{proof}
        The trace distance of the single-mode cubic phase gate given by $f(\hat q)=\gamma \hat q^3$ can be bounded by
        \begin{align}
			\frac 1 2\left\| \hat P_d^f \text{Tr}_S (\hat\rho) \hat P_d^{f \dagger} - \text{Tr}_S (\hat P^f \hat\rho \hat P^{f \dagger}) \right\|_1 
			\leq &\frac{1}{\ell} \int_{\mathbb T_{\ell}^{2}} \text{d} \mathbf{t} \left\|  [\hat \rho,\gamma \hat q^3 - \gamma (\hat q-t_q)^3] \right\|_1\nonumber\\
            \leq&\frac{3}{\ell} \int_{\mathbb T_{\ell}^{2}} \text{d} \mathbf{t} \left\|  [\hat \rho,\gamma t_q^2\hat q] \right\|_1+\left\|  [\hat \rho,\gamma t_q\hat q^2] \right\|_1\nonumber\\
            \leq&\frac{3 |\gamma|}{\ell} \int_{\mathbb T_{\ell}^{2}} \text{d} \mathbf{t} \left| t_q^2 \right|\left\|  [\hat \rho, \hat q] \right\|_1 + \left| t_q \right|\left\|  [\hat \rho, \hat q^2] \right\|_1\nonumber\\
            \leq&\frac{3 |\gamma|}{\ell} \int_{\mathbb T_{\ell}^{2}} \text{d} \mathbf{t} \left| t_q^2 \right|\left\|  [\hat \rho, \hat q] \right\|_1 + \left| t_q \right| \left\|  [\hat \rho, \hat q] \right\|_1 \left(d +1 \right)\ell\nonumber\\
			\leq&\frac{3|\gamma| \sqrt{2E_{\hat \rho}}}{\ell} \int_{\mathbb T_{\ell}^{2}} \text{d} \mathbf{t} |t_q^2|+\left(d +1 \right)\ell |t_q|  \nonumber\\
			\leq&\frac{3|\gamma| \sqrt{2E_{\hat \rho}}}{\ell} \int_{\mathbb T_{\ell}^{2}} \text{d} \mathbf{t} |t_q^2|+\frac 3 2 d\ell |t_q|  \nonumber\\
			=&3|\gamma| \sqrt{2E_{\hat \rho}} \int_{\mathbb T_{\ell}} \text{d} t_q |t_q^2|+\frac 3 2 d\ell|t_q| \nonumber\\
			=&3|\gamma| \sqrt{2E_{\hat \rho}}\left(\frac{\ell^3}{12}+3d\frac{\ell^3}{8}\right)\nonumber\\
            \leq&5|\gamma|\sqrt{2E_{\hat \rho}}d\ell^3/4\nonumber\\
            =&5|\gamma|\pi^{3/2}\sqrt{E_{\hat \rho}/d},
		\end{align}
        where we have used $d\geq2$ and the fact that in the CCVQC model $\left\|[\hat \rho, \hat q^2] \right\|_1 = \left\|\hat \Lambda_d[\hat \rho, \hat q^2] \hat \Lambda_d\right\|_1 \leq 2 \left\|[\hat \rho, \hat q] \right\|_1 \left\|\hat \Lambda_d \hat q \right\|_\infty \leq \left\|[\hat \rho, \hat q] \right\|_1 \left( d +1 \right)\ell$.
        \end{proof}
        Assuming a maximum value of $\gamma\leq 8E^{*3/2}$ according to Lemma~\ref{lemma:energy-cubic-phase}, we find
        \begin{align}
			\frac 1 2\left\| \hat P_d^f \text{Tr}_S (\hat\rho) \hat P_d^{f \dagger} - \text{Tr}_S (\hat P^f \hat\rho \hat P^{f \dagger}) \right\|_1 
			\leq& 40 \pi^{3/2} E_{\hat \rho}^{*2}/\sqrt d \nonumber \\
			\leq&223 E^{*2}/\sqrt{d}.
		\end{align}
        \subsection{Proof of error of controlled-\texorpdfstring{$Z$}{Z} gate}
        \label{sec:appendix-error-controlled-z}
        Here we provide a proof of Lemma~\ref{lemma:error-controlled-z}.
        \begin{proof}
        A simple two-mode interaction is provided by the controlled-$Z$ gate with $f(\hat q_k, \hat q_l)=s \hat q_k\hat q_l$ bounded by
        \begin{align}
			\frac 1 2\left\| \hat P_d^f \text{Tr}_S (\hat\rho) \hat P_d^{f \dagger} - \text{Tr}_S (\hat P^f \hat\rho \hat P^{f \dagger}) \right\|_1 
			&\leq\frac{1}{2\ell^2 } \int_{\mathbb T_{\ell}^{4}} \text{d} \mathbf{t} \left\|  [\hat \rho, s t_{q_j}\hat q_i] \right\|_1+\left\|  [\hat \rho,s t_{q_k}\hat q_l] \right\|_1\notag\\
			&=\frac{1}{2\ell^2 } |s| \int_{\mathbb T_{\ell}^{4}} \text{d} \mathbf{t} t_{q_k} \left\|  [\hat \rho, \hat q_k] \right\|_1+t_{q_l}\left\|  [\hat \rho,\hat q_l] \right\|_1\notag\\
			&=\frac{1}{2\ell^2 } |s| 2\sqrt{2E}\int_{\mathbb T_{\ell}^{4}} \text{d} \mathbf{t} \,(t_{q_l} +t_{q_k})\notag\\
			&=\frac{1}{2\ell^2 } |s| 2\sqrt{2E}\ell^3\int_{\mathbb T_{\ell}} \text{d} t' \,2|t'|\notag\\
			&=\frac{1}{2\ell^2 } |s| 2\sqrt{2E}\ell^32\frac{\ell^2}{4}\notag\\
			&= 2|s| \sqrt{E}(\pi/d)^{3/2}
		\end{align}
		\end{proof}
\subsection{Proof of error of rotation}
\label{sec:appendix-rotation} 
Here we provide a proof of Lemma~\ref{lemma:error-rotation}.
\begin{proof}
    We consider a rotation of an arbitrary angle $\theta$. Note that without loss of generality, we can restrict to the case that $\theta\in [-\frac{\pi}{4},\frac{\pi}{4})$ because we can generate rotations that are multiples of $\pi/2$ using the Fourier transform. We can decompose $\hat R(\theta)$, from Eq.~(\ref{eq: rotation gate}), into Fourier transforms and shear gates~\cite{douce2019} as 
    \begin{align}
        \label{eq:decomposition-fourier}
        \hat{R}(\theta)=\hat F\hat P(s_3)\hat F\hat P(s_2)\hat F\hat P(s_1)
    \end{align}
    whereby $s_1=\sec\theta+\tan\theta$, $s_2=\cos\theta$ and $s_3=\cos\theta+(1+\sin\theta)\tan\theta$. We note that given the restrictions on the angles, we have that $\sqrt{2}-1\leq s_1 \leq \sqrt{2}+1$, $\frac{1}{\sqrt{2}} \leq s_2\leq 1$, and $\sqrt{2}-1\leq s_3\leq \sqrt{2}+1$.

    Also note that although the total energy of the state after each of the three shear gates will be unchanged, the decomposition does not guarantee that the energy will not change during the application of each shear gate. Therefore, when calculating the error introduced by the second and third shear gates, we must consider the maximum possible energy of the state after the first and second shear gates, respectively.

    We calculate the trace distance of applying the gate
    \begin{align}
        &\frac 1 2 \|\Tr_S(\hat R(\theta)\hat \rho \hat R^\dagger(\theta))-\hat R_d(\theta)\Tr_S(\hat \rho)\hat R_d^\dagger(\theta)\|_1\nonumber\\
        =&\frac 1 2\|\Tr_S(\hat P(s_3) \hat \rho'' \hat P^\dagger(s_3))-\PhaseGate[s_3]\hat F_d \PhaseGate[s_2]\hat F_d \PhaseGate[s_1]\Tr_S(\hat \rho)\PhaseGate[s_1]^\dagger \hat F_d^\dagger\PhaseGate[s_2]^\dagger \hat F_d^\dagger\PhaseGate[s_3]^\dagger \|_1\nonumber\\
        \leq&\frac 1 2\|\Tr_S(\hat P(s_1)\hat \rho  \hat P^\dagger(s_1))-\PhaseGate[s_1]\Tr_S(\hat \rho)\PhaseGate[s_1]^\dagger \|_1\nonumber\\
        &+\frac 1 2\|\Tr_S(\hat P(s_2)\hat \rho'  \hat P^\dagger(s_2))-\PhaseGate[s_2]\Tr_S(\hat \rho')\PhaseGate[s_2]^\dagger \|_1\nonumber\\
        &+\frac 1 2\|\Tr_S(\hat P(s_3)\hat \rho''  \hat P^\dagger(s_3))-\PhaseGate[s_3]\Tr_S(\hat \rho'')\PhaseGate[s_3]^\dagger \|_1
    \end{align}
where we used Lemma~\ref{lemma:error-Fourier-transform}. Inspecting each of these terms, we can bound the total error of the gate. The first term is simply the error of applying a shear gate to a state $\hat\rho$ with total energy $E_{\hat \rho}$. Using Lemma~\ref{lemma:error-shear-gate}, we have
\begin{align}
    \frac 1 2\|\Tr_S(\hat P(s_1)\hat \rho  \hat P^\dagger(s_1))-\PhaseGate[s_1]\Tr_S(\hat \rho)\PhaseGate[s_1]^\dagger \|_1\leq |s_1|\frac{\pi}{d}\sqrt{E_{\hat \rho}} /\sqrt 2.
\end{align}
Next, we note that a bound on the energy of the state $\hat \rho'=\hat F \PhaseGate[s_1] \hat \rho\PhaseGate[s_1]^\dagger \hat F^\dagger$ can be calculated using Lemma \ref{lemma:energy-shear}. We have $E_{\hat \rho'}\leq (1+s_1)^2E_{\hat \rho}$. Therefore,
\begin{align}
    \frac 1 2\|\Tr_S(\hat P(s_2)\hat \rho'  \hat P^\dagger(s_2))-\PhaseGate[s_2]\Tr_S(\hat \rho')\PhaseGate[s_2]^\dagger \|_1\leq& |s_1|\frac{\pi}{d}  \sqrt{ E_{\hat\rho'}/2}\nonumber\\
    \leq& |s_2|(1+|s_1|)\frac{\pi}{d} \sqrt{ E_{\hat\rho}/2}.
\end{align}
Furthermore, a bound on the energy of the state $\hat \rho''$ can be calculated in terms of $E_{\hat \rho''}\leq (1+s_2)^2(1+s_1)^2E_{\hat \rho}$, hence we have 
\begin{align}
    \frac 1 2 \|\Tr_S(\hat P(s_3)\hat \rho'  \hat P^\dagger(s_3))-\PhaseGate[s_3]\Tr_S(\hat \rho')\PhaseGate[s_3]^\dagger \|_1\leq& |s_3|\frac{\pi}{d} \sqrt{E_{\hat\rho''}/2}\nonumber\\
    \leq& \frac{\pi}{d} \sqrt{ E_{\hat\rho}/2}\left(1+|s_1|\right)\left(1+|s_2| \right)|s_3|.
\end{align}
Also, note that we can maximize over all possible values of $\theta$ and find that
\begin{align}
    |s_1|+(1+|s_1|)|s_2|+(1+|s_1|)(1+|s_2|)|s_3|< 23.
\end{align}
Therefore, combining these three bounds, we find
 \begin{align}
        \frac{1}{2}\|\Tr_S(\hat R(\theta)\hat \rho \hat R^\dagger(\theta))-\hat R_d(\theta)\Tr_S(\hat \rho)\hat R_d^\dagger(\theta)\|_1\leq & 23\frac \pi d \sqrt{E_{\hat\rho}/2}\nonumber\\
        \leq & 52\frac{\sqrt{E_{\hat \rho}}}{d}.
    \end{align}
\end{proof}
\subsection{Proof of error of Mach-Zehnder interferometer}
		\label{sec:appendix-beamsplitter}
        Here we provide the proof of Lemma~\ref{lemma:error-beamsplitter}.
        \begin{proof}
		We consider a beamsplitter of arbitrary angle $\theta$. Note that without loss of generality, we can restrict to the case that $\theta\in [-\frac{\pi}{4},\frac{\pi}{4})$ because we can cover all other angles using swap gates, which contribute no error in the trace distance.
		
		We will use the following decomposition of a beamsplitter into controlled-$Z$ gates, Fourier transforms, and a rotation
        \begin{align}
			\BSGate*(\theta,\phi) = \CZGate*(s_1)\hat{F}_k^\dagger F_l \CZGate*(s_2)\hat F_k^\dagger\hat F_l \CZGate*(s_1)\hat{F}_k^\dagger\hat{F}_l\hat R_k(\phi)
		\end{align}
		whereby $s_1=\sec(\theta)(1-\sin(\theta))$ and $s_2=-\cos(\theta)$. We note that given the restrictions on the angles, we have that $1-\sqrt{2}\leq s_1 \leq \sqrt{2}-1$ and $-\frac{1}{\sqrt{2}}\leq s_2\leq \frac{1}{\sqrt{2}}$.
		With the trace distance for a controlled-$Z$ gate given by Lemma~\ref{lemma:error-controlled-z}
		\begin{align}
			2 |s| \sqrt{E}(\pi/d)^{3/2},
		\end{align}
		and given that, as specified by Lemma~\ref{lemma:energy-controlled-Z}, the maximum energy of the new state is upper bounded by
		\begin{align}
			E_{\hat \rho'}\leq\left(1+s\right)^4E_{\hat \rho},
		\end{align}
		we can follow the same steps as for the proof of rotation. We see that
		\begin{align}
			|s_1|+(1+s_1)^2|s_2|+(1+s_1)^2(1+s_2)^2|s_1|< 5
		\end{align}
		and we find
		\begin{align}
			\frac{1}{2}\left\|\Tr_S(\BSGate*(\theta,0)\hat \rho \BSGate*^\dagger(\theta,0))-\BSGate(\theta,0)\Tr_S(\hat \rho){\BSGate}^\dagger(\theta,0)\right\|_1
			\leq  56\sqrt{\frac{E_{\hat \rho}}{d^3 }}.
		\end{align}
        When including the variable $\phi$, we simply add the error of a single-mode rotation to find
		\begin{align}
			\frac{1}{2}\left\|\Tr_S(\BSGate*(\theta,\phi)\hat \rho \BSGate*^\dagger(\theta,\phi))-\BSGate(\theta,\phi)\Tr_S(\hat \rho){\BSGate}^\dagger(\theta,\phi)\right\|_1
			\leq  &56\sqrt{\frac{E_{\hat \rho}}{d^3 }}+ 52\frac{\sqrt{E_{\hat \rho}}}{d }\nonumber\\
			\leq & 108\frac{\sqrt{E_{\hat \rho}}}{d}.
		\end{align}
        \end{proof}

\subsection{Proof of error of squeezing}
\label{sec:appendix-squeezing}
Here we provide a proof of Lemma~\ref{lemma:error-squeezing}.
\begin{proof}
    We will use the following decomposition of squeezing into Fourier transforms and shear gates~\cite{douce2019}
    \begin{align}
        \hat S(r)=\hat F\hat P(s_3)\hat F\hat P(s_2)\hat F\hat P(s_1)
    \end{align}
    whereby $s_1=s_3=e^r$ and $s_2=e^{-r}$. Note that this is the same decomposition as the one given in Eq.~(\ref{eq:decomposition-fourier}) for rotation, except with different parameters. We note that given $r \geq 0$, the trace distance for a shear gate given by Lemma~\ref{lemma:error-shear-gate}
		\begin{align}
			|s| \frac{\pi}{d} \sqrt{\frac{E_{\hat \rho}}{2}},
		\end{align}
		and given that, as specified by Lemma~\ref{lemma:energy-shear}, the maximum energy of the new state is upper bounded by
		\begin{align}
			E_{\hat \rho'}\leq\left(1+s\right)^2E_{\hat \rho},
		\end{align}
		we can follow the same steps as for the proof of rotation. We see that
		\begin{align}
			&|s_1|+(1+s_1)|s_2|+(1+s_1)(1+s_2)|s_3| \nonumber \\
            =&|e^r|+(1+e^r)|e^{-r}|+(1+e^r)(1+e^{-r})|e^r| \nonumber \\
            =& 2+e^{-r}+3e^r+e^{2r} \nonumber \\
            \leq & 7 e^{2r},
		\end{align}
		and therefore
		\begin{align}
			\frac 1 2 \left\|\Tr_S(\hat S(r)\hat \rho \hat S^\dagger(r))-\hat S_d(r)\Tr_S(\hat \rho)\hat S_d^\dagger(r)\right\|_1 \leq  7e^{2r}\frac{\pi}{d} \sqrt{\frac{E_{\hat \rho}}{2}}.
		\end{align}
\end{proof}
Also from Eq.~\eqref{eq:max squeezing}, we can note that it is possible to maximize over all possible values of $\frac{1}{\sqrt{2E^*}}\leq e^{r}\leq \sqrt{2E^*}$ and find that
\begin{align}
    |e^r|+(1+e^r)|e^{-r}|+(1+e^r)(1+e^{-r})|e^r|=& 2+e^{-r}+3e^r+e^{2r}\nonumber\\
    \leq & 2+4\sqrt{2E^*}+2E^*\nonumber\\
    \leq & 14E^*,
\end{align}
where in the last line we use that $E^*\geq 1/2$.
Therefore, combining these bounds, we find
 \begin{align}
        \frac 1 2 \left\|\Tr_S(\hat S(s)\hat \rho \hat S^\dagger(s))-\hat S_d(s)\Tr_S(\hat \rho)\hat S_d^\dagger(s) \right\|_1 \leq & 14 E^* \frac \pi d \sqrt{\frac {E_{\hat \rho}}{2}}\nonumber\\
        \leq & 14 \frac \pi d \frac {E^{*3/2}}{\sqrt{2}} \nonumber \\
        \leq & 32 \frac{E^{*3/2}}{d}.
\end{align}
    
\section{Measurements are equivalent to modular measurements when restricting to a cut-off in position}
\label{appendix:restriction-mod}
In this Appendix, we show more rigorously that $\hat\Lambda_d\hat K_{\bar{\mathbf x}}\hat\Lambda_d=\hat{\tilde K}_{\bar{\mathbf x}}$, which is used in the proof of Lemma~\ref{lemma:MCVQC-CCVQC} in the main text. Note that it is sufficient to prove this fact for a single mode, as the multimode case is a trivial generalization of this case. We have
\begin{align}
    \hat\Lambda_d=\int_{-(\lfloor d/2\rfloor+\frac 1 2)\ell}^{(\lceil d/2\rceil-\frac 1 2)\ell} \, \text{d} x \, \ketbra{\hat{q}=x}{\hat{q}=x}
\end{align}
and
\begin{align}
\hat K_{\bar x}=&\int_{-\ell/2}^{\ell/2} \dd s \sum_{m \in \mathbb Z}\ketbra{\hat{q}=\bar x+\ell d  m+s}{\hat{q}=\bar x+\ell d  m+s},
\end{align}
where $\bar x=\ell j$ and $j\in \mathbb Z_d$. 
Hence, we have
\begin{align}
    \hat\Lambda_d\hat K_{\bar x}=&\int_{-(\lfloor d/2\rfloor+\frac 1 2)\ell}^{(\lceil d/2\rceil-\frac 1 2)\ell} \, \text{d} x \, \int_{-\ell/2}^{\ell/2} \dd s \sum_{m \in \mathbb Z}\ketbra{\hat{q}=x}{\hat{q}=x}\ketbra{\hat{q}=\ell (dm+j)+s}{\hat{q}=\ell (dm+j)+s}.
\end{align}
Here we note that the bra-ket term is zero unless $m=0$. We are therefore left with
\begin{align}
   \hat\Lambda_d\hat K_{\bar x}=&\int_{-(\lfloor d/2\rfloor+\frac 1 2)\ell}^{(\lceil d/2\rceil-\frac 1 2)\ell} \, \text{d} x \, \int_{-\ell/2}^{\ell/2} \dd s \ketbra{\hat{q}=x}{\hat{q}=x}\ketbra{\hat{q}=\ell j+s}{\hat{q}=\ell j+s}\nonumber\\
    =& \int_{-\ell/2}^{\ell/2} \dd s \ketbra{\hat{q}=\ell j+s}{\hat{q}=\ell j+s}\nonumber\\
    =& \hat{\tilde K}_{\bar x}.
\end{align}
\end{widetext}
\end{document}